\documentclass{article}[12pt]

\usepackage{amsmath,amsfonts,amssymb,amsthm}
\usepackage{fullpage}
\usepackage{pdfsync}
\usepackage{hyperref}
\usepackage{microtype}
\usepackage{enumitem}
\usepackage{graphics,graphicx,color}
\usepackage{multirow}
\usepackage{mathrsfs}
\usepackage{subfig}
\usepackage{float}
\usepackage{paralist}

\usepackage{tabulary}
\usepackage{appendix}
\usepackage{thmtools}
\usepackage{thm-restate}
\usepackage{cleveref}
\usepackage{bbm}
\usepackage{wrapfig}

\newcommand{\Lap}{\text{Lap}}

\newcommand{\Var}{\mathrm{Var}}

\newcommand{\E}[1]{\mathbb{E}\left[#1\right]}
\newcommand{\abs}[1]{\left\lvert{#1}\right\rvert}

\newcommand{\eps}{\epsilon}

\newcommand{\pufQ}{\mathcal{Q}}

\newcommand{\inc}{\mathsf{i}}
\newcommand{\gen}{\mathsf{g}}
\newcommand{\sat}{\mathsf{s}}
\newcommand{\hgt}{\mathsf{h}}
\newcommand{\wgt}{\mathsf{w}}

\newcommand{\parent}{\mathsf{parent}}

\newcommand{\dbprivacy}{dataset attribute privacy}

\usepackage{xcolor}
\definecolor{DarkGreen}{rgb}{0.1,0.5,0.1}

\newcommand{\olya}[1]{\textcolor{blue}{[Olya: #1]}}

\newcommand{\fixme}[1]{\textcolor{orange}{#1}}

\usepackage{algorithmicx,algpseudocode, algorithm}

\newtheorem{theorem}{Theorem}

\newtheorem{definition}{Definition}

\newtheorem{example}{Example}

\title{Attribute Privacy: Framework and Mechanisms}
\author{Wanrong Zhang\footnotemark[1] \and Olga Ohrimenko\footnotemark[2] \and Rachel Cummings\footnotemark[3]}

\begin{document}

\maketitle

\renewcommand{\thefootnote}{\fnsymbol{footnote}}
\footnotetext[1]{School of Industrial and Systems Engineering, Georgia Institute of Technology. Email: \texttt{wanrongz@gatech.edu}. Supported in part by a Mozilla Research Grant, NSF grant CNS-1850187, and an ARC-TRIAD Fellowship from the Georgia Institute of Technology. Part of this work was completed while the author was at Microsoft Research.}
\footnotetext[2]{School of Computing and Information Systems, The University of Melbourne. Email: \texttt{oohrimenko@unimelb.edu.au}. Part of this work was completed while the author was at Microsoft Research.}
\footnotetext[3]{Department of Industrial Engineering and Operations Research, Columbia University. Email: \texttt{rac2239@columbia.edu}. Supported in part by a Mozilla Research Grant, a Google Research Fellowship, NSF grants CNS-1850187 and CNS-1942772 (CAREER), and a JPMorgan Chase Faculty Research Award. Most of this work was completed while the author was at Georgia Institute of Technology.}

\renewcommand{\thefootnote}{\arabic{footnote}}

\begin{abstract}
%Privacy of training data is a growing concern given that many machine learning models are trained on confidential and potentially sensitive data. These concerns continue to rise with a growing number of attacks showing that leakage of individual records and dataset properties is possible. Differentially private mechanisms aim to alleviate these concerns by protecting leakage of individual data. However, global properties about a dataset can still be revealed posing privacy concerns if these properties are sensitive. 
%
%In this work, we depart from individual privacy to initiate the study of attribute privacy, where a data owner is concerned about revealing sensitive properties of a whole dataset
%(e.g., mortality rate in a hospital as rather than 
%presence or absence of an individual).
%We propose definitions
%to capture attribute privacy in two different settings: when global properties of the dataset are sensitive, and when the underlying distribution that generated the data is sensitive. 

Ensuring the privacy of training data is a growing concern since many machine learning models are trained on confidential and potentially sensitive data. Much attention has been devoted to methods for protecting individual privacy during analyses of large datasets. However in many settings, global properties of the dataset may also be sensitive (e.g., mortality rate in a hospital rather than presence of a particular patient in the dataset). In this work, we depart from individual privacy to initiate the study of attribute privacy, where a data owner is concerned about revealing sensitive properties of a whole dataset during analysis. We propose definitions to capture \emph{attribute privacy} in two relevant cases where global attributes may need to be protected: (1) properties of a specific dataset and (2) parameters of the underlying distribution from which dataset is sampled. We also provide two efficient mechanisms and one inefficient mechanism that satisfy attribute privacy for these settings. We base our results on a novel use of the Pufferfish framework to account for correlations across attributes in the data, thus addressing ``the challenging problem of developing Pufferfish instantiations and algorithms for general aggregate secrets'' that was left open by \cite{kifer2014pufferfish}.
\end{abstract}

%!TEX root = sample-authordraft.tex

\section{Introduction}

Privacy in the computer science literature has generally been defined at the \emph{individual level}, such as differential privacy \cite{DMNS06}, which protects the value of an individual's data within analysis of a larger dataset.
However, there are many settings where confidential information contained in the data goes beyond presence
or absence of an individual in the data and instead relates to attributes at the \emph{dataset level}.
Global properties about attributes revealed from data analysis may leak trade secrets,
intellectual property and other valuable information pertaining to the data owner,
even if differential privacy is applied~\cite{10.1145/2020408.2020598}.

In this paper, we are interested in privacy of \emph{attributes in a dataset}, where an analyst must prevent \emph{global properties} of sensitive attributes in her dataset from leaking during analysis. For example, insurance quotes generated by a machine-learned model might leak information about how many female and male drivers are insured by the company that trained the model; voice and facial recognition models may leak the distribution of race and gender among users in the training dataset \cite{10.1504/IJSN.2015.071829, BG18}.
 Under certain circumstances, even releasing the distribution from which the data were sampled may be sensitive. For example, experimental findings by a pharmaceutical company measuring the efficacy of a new drug would be considered proprietary information. 
It is important to note that the problem we consider here departs from individual-level attribute privacy where
one wishes to protect attribute value of a record (e.g., person's race)
as opposed to a function over all values of this attribute in the dataset (e.g., race distribution in a dataset).
%\newtext{It is important to note that protecting an attribute value (e.g., gender or race) of an \emph{individual} is a separate problem and can be addressed, for example, with differential privacy.}

%For example, a university may want to announce the average SAT score of its incoming class, but because family income is known to be correlated with standardized test scores, this may additionally reveal (protected) financial information about the student population.

{Several recent attacks show that global properties about a dataset can indeed be leaked
from machine learning model APIs~\cite{Song2020Overlearning,DBLP:journals/corr/abs-1805-04049,zhang2020datasetlevel,10.1145/3243734.3243834}.
In fact, these works show that models learn sensitive attributes even when censorship is applied or when
the attributes are deemed irrelevant for the actual learning task.
Hence, the naive solution of removing sensitive attributes from the dataset
is insufficient, as attributes are often correlated, and protected information can still be leaked
by releasing non-sensitive information about the data.
Though differential privacy can be used to protect sensitive attributes at the individual level
(e.g., in the algorithmic fairness literature~\cite{DHP+12}),
the study of attribute privacy at the dataset or distribution level is limited, both in terms of a framework for reasoning about it and mechanisms for protecting it.}

\subsection{Our Contributions}

\if 0
\fixme{The contributions of this work are three-fold.
We first identify the gap in the literature for preserving the privacy of global properties
about a dataset. We then propose definitions to allows an analyst to capture information
that they wish to protect at a dataset level.
Since protecting global properties about a
dataset can be expensive in general,
we identify settings that are of practical interest
and propose efficient mechanisms for satisfying our definitions.
We highlight these contributions in more detail below.}
\fi

%The goal of this paper is to 
\paragraph{Problem formulation} We
initiate the study
of \emph{attribute privacy} at the dataset and distribution level and establish the first formal framework
for reasoning about these privacy notions.
We identify two cases
where information about global properties of a dataset may need to be protected:
 (1) properties of a specific dataset
and (2) parameters of the underlying distribution from which dataset is sampled.
We refer to the first setting as \emph{dataset attribute privacy}, where the data owner wishes to protect properties of her sample from a distribution, but is not concerned about revealing the distribution. For example, even though the overall prevalence of a disease may be known, a hospital may wish to protect the fraction of its patients with that disease. 
We refer to the second setting as \emph{distributional attribute privacy}, which considers the distribution parameter itself a secret. For example, demographic information of the population targeted by a company may reveal information about its proprietary marketing strategy.
These two definitions distinguish between protecting a sample and protecting the distribution from which the dataset is sampled.

%The former considers protecting properties
%of a realization (or a sample) of a distribution 
%where revealing the underlying distribution may not be of concern
%may not be of concern (for example,
%it could already be known).
%Here, the data owner wishes to protect its ``sample''.
%%as it may be conditional on some latent variable such as a geographical reason.
%We refer to this as \emph{dataset attribute privacy}.
%The second case applies to scenarios
%where the distribution parameter
%is itself a secret.
%For example,
%it would reveal the demographic information
%of company's customers.
%We refer to this as \emph{distributional attribute privacy.}
%The distinction between the two stems from protecting a sample from a known
%distribution vs.~protecting the distribution of where the dataset is sampled from.

\paragraph{Definitions of Attribute Privacy.} We propose definitions for capturing dataset and distributional attribute privacy by instantiating a general privacy framework called the Pufferfish framework~\cite{kifer2014pufferfish}. This framework was originally introduced to handle correlations across individual entries in a database.
Instantiating this framework for attribute privacy is non-trivial as it requires
reasoning about secrets and parameters at a dataset level.

For dataset attribute privacy, our definition considers the setting where individual records are independent of each other
while correlations may exist between attribute values of each record.
%To this end, we show that the framework can also be used to capture correlations across database-level attributes.
Then, to be able to capture general global properties of a dataset that need to be protected, we choose to express secrets
as functions over attribute values across all records in a dataset.
For example, this allows one to express that the average income of individuals in a dataset
being below or above \$50K is secret information.

Our second definition also instantiates the Pufferfish framework
while explicitly capturing the random variables used to generate
attribute values of a record. Here, the parameters of the distribution
of protected attributes are treated as confidential information.
For example, in a dataset where records capture trials
in a stochastic chemical environment, 
one can express that determining whether the probability with which a certain compound
is added in each trial is 0.2 or 0.8 is a secret.

\paragraph{Mechanisms to Protect Attribute Privacy.}
Our definitions allow
an analyst to specify secrets about global properties of a dataset
that they wish to protect. In order to satisfy these definitions
the analyst can use a general tool for providing Pufferfish privacy
called the Wasserstein mechanism proposed by Song~\textit{et al.}~\cite{song2017pufferfish}.
However, this mechanism is computationally expensive and may require computing an exponential number of pairwise Wasserstein distances, which is not feasible in most practical settings.  To this end, we propose efficient mechanisms in the following two settings.

%Satisfying the definitions
%is challenging
%as one needs to capture how a change in values of one attribute
%for all records in the database
%affects other attributes and, as a consequence, how it affects the output
%of a function computed on this data (i.e., ``sensitivity'' of a function).
%Computing the sensitivity for all possible changes is also computationally expensive.
%In fact, all our definitions can be satisfied using a general mechanism
%for satisfying the Pufferfish definition
%called the Wasserstein mechanism from~\cite{song2017pufferfish}.
%Though general, this mechanism is computationally
%expensive.
%To this end, we propose efficient mechanisms
%in the following settings:
%It is crucial to understand that making an assumption
%about the distribution helps
%capture attribute effect in a constructive manner
%and provide worst-case guarantees for the mechanisms.
%Without any assumptions,
%the database may change arbitrarily when all values of one attribute change,
%making any privacy-preserving mechanism futile in terms of utility.
%\begin{compactitem}
%\item
For dataset attribute privacy, we consider a special class
of functions and attribute properties and
propose a mechanism based on Gaussian noise.
Though the nature of the noise is added from the same family of distributions as
the differentially private Gaussian Mechanism, in Section~\ref{sec:datasecrets} we articulate that
the similarity between the two is based solely on the nature of the noise.
In Section~\ref{sec:gaus_precise}, we show that the mechanism can be applied to datasets where (1) attributes follow a multivariate Gaussian distribution
and (2) the function to be computed on the data and the attribute property to be protected are linear in the number of records in the dataset (e.g., mean).
%even in the absence of such strong assumptions.
%\olya{Should we say a bit more? This sentence sounds like an afterthought} 
%\rc{added a little more}
We note that with the help of
variational auto-encoders (VAEs)~\cite{DBLP:journals/corr/KingmaW13}, one
can obtain a Gaussian representation of the data
even if a dataset does not come from a Gaussian distribution.
Moreover, such disentangled representations
can be based on interpretable attributes~\cite{Higgins2017betaVAELB}
that are easier for specifying which attributes require protection, particularly when the original data are complex (e.g., pixels on an image vs. the gender of the person in it).\footnote{Though naive use of VAEs may not provide end-to-end privacy guarantees, it serves as an example that it is possible to obtain a representation of non-Gaussian data
with interpretable Gaussian features. We leave it as an interesting open question on how to provide end-to-end privacy-preserving feature disentanglement.}
 Nevertheless, we also consider the case
where data may not follow Gaussian distribution.
Specifically, in Section \ref{s.nogauss},
we relax the Gaussian assumption 
and show that our mechanism can still provide dataset attribute privacy by leveraging
Gaussian approximations.

%This would then allow an analyst to specify which latent features are deemed sensitive for the data where the original features are less descriptive (e.g., pixels on an image vs. the gender of the person in it).

%\item
For distributional attribute privacy,
we consider a model where dependencies between the attributes
form a Bayesian network.
This model helps us capture the extent to which
a sensitive attribute parameter affects
parameters of attributes in the query,
and we add noise proportional to this influence.
Although our mechanism is inspired by the Markov Quilt mechanism ~\cite{song2017pufferfish},
the difference in settings prompts several changes,
including a different metric for measuring influence between the variables.
%\end{compactitem}

Finally, we note that
although~\cite{kifer2014pufferfish} identified
that
``there is little focus in the literature on rigorous
and formal privacy guarantees for business data'',
they leave ``the challenging
problem of developing Pufferfish instantiations and algorithms for general aggregate
secrets" as future work.

%!TEX root = attributePrivacy.tex

%In the introduction we highlight some of the attacks that motivate our work
%and focus this section
%on privacy notions and mechanisms.

\subsection{Related work} 

%\rc{cite your new paper somewhere in this section!}
%\oo{done in intro}

Machine learning models have been shown
to memorize and leak data used to train them,
raising questions about
the release and use of these models
in practice.
For example, membership attacks~\cite{DBLP:conf/sp/ShokriSSS17}
show that models can leak whether certain
records (e.g., patient data) were part of the training dataset or not.
Attribute (or feature) privacy attacks, on the other hand,
consider leakage of attribute values at
an individual level~\cite{Song2020Overlearning,
DBLP:conf/ccs/FredriksonJR15}, and property inference attacks show that global properties
about datasets can be leaked~\cite{DBLP:journals/corr/abs-1805-04049,10.1145/3243734.3243834,10.1504/IJSN.2015.071829}.
%While attacks on record membership or leakage of attribute values of a record
%can be addressed
%with differential privacy (DP)~\cite{abadi,privacybook}.
%the study of attribute privacy at a dataset or distribution level
%is limited both in terms of a framework
%for reasoning about it as well as
%mechanisms for protecting it.

Differential privacy (DP) \cite{DMNS06,privacybook} guarantees individual-level privacy when publishing an output computed on a database, by bounding the influence of any record on the output and adding noise. Importantly, DP does not aim to protect population-level information, and was designed to learn global properties of a dataset without sacrificing individual privacy.
DP does provide \emph{group privacy} guarantees for groups of $k$ correlated records, but these quantitative guarantees are only meaningful when $k$ is small relative to the size of the dataset. Syntactically, DP guarantees that if any individual record were to be changed---including all attributes of that record---the result of the analysis would be approximately the same. For attribute privacy, we seek similar guarantees if an entire attribute of the dataset were to be changed---including all individuals' values for that attribute.

%(e.g.. see~\cite{privacybook,10.1145/2020408.2020598} for specific examples).
%In fact, one of the goals of DP is to be able to learn properties and trends of a dataset
%without sacrificing individual privacy.
%A related notion is group differential privacy~\cite{privacybook}
%which extends the definition to account for privacy loss
%incurred by a group of $k$ records. Since this notion
%assumes correlation between all records in a group,
%protecting privacy of attribute values for all records
%would results in the addition of noise proportional to the size
%of the database which will destroy the utility.
%Though one could potentially set $k$ to be
%the number of times a particular attribute value appears,
%how to set $k$ while preserving privacy of the number
%of records with a specific attribute value vs.~preserving utility is not clear.
The Pufferfish framework~\cite{kifer2014pufferfish} that we instantiate
and describe in detail in the following sections, can be seen as
a generalization of differential privacy that 
explicitly states the information that needs to be kept secret and the
adversary's background knowledge about the data.
Blowfish privacy~\cite{sigmod:HeMD14}
also allows one to express
secrets and publicly known information about the data,
but expressed 
as
constraints on the data rather than distributions over data.
We adapt the Markov Quilt Mechanism from~\cite{song2017pufferfish}, who also employ the Pufferfish framework~\cite{kifer2014pufferfish} for private analysis of correlated data, although they focus on individual-level privacy.  Our focus instead on privacy of dataset properties and distributions leads to a substantially different instantiation of the Pufferfish framework where the secrets are defined over attribute values rather than individual records in the dataset.

Research on algorithmic fairness has proposed several definitions formalizing the idea that machine learning models should not
exhibit discrimination based on protected attributes (e.g., gender or race). 
Demographic parity formalizes fairness
by requiring that a classifier's predicted label
is independent of an individual's protected attributes. Our notion of dataset attribute privacy is a general framework where one can specify
what information about attributes need to be protected,
with attribute independence being one such scenario.
However, our attribute privacy definitions would not be useful for satisfying other fairness notions that explicitly incorporate protected attributes, such as affirmative action or fairness through awareness \cite{DHP+12}.
Moreover, techniques proposed
to obtain fair representations of the training
data~\cite{DBLP:journals/corr/LouizosSLWZ15,pmlr-v28-zemel13}
have been shown to still leak
sensitive attributes~\cite{Song2020Overlearning} when applied in the privacy context.

%Research on algorithmic fairness~\cite{DHP+12,pmlr-v28-zemel13} aims
%to train machine learning models that do not
%exhibit discrimination based on protected attributes (e.g., gender or race).
%Demographic parity formalizes fairness
%by ensuring that the prediction of, for example, a classifier
%is independent of a protected attribute or attributes.
%Database attribute privacy, that we consider in this paper,
%is a general framework where one can specify
%what information about attributes need to be protected
%with attribute independence being one scenario.
%Our definition also explicitly states background knowledge of the adversary
%which may not be relevant in the context of fairness.
%Moreover, techniques proposed
%to obtain fair representations of the training
%data~\cite{DBLP:journals/corr/LouizosSLWZ15,pmlr-v28-zemel13}
%have been shown to still leak
%sensitive attributes~\cite{Song2020Overlearning} when applied in the privacy context
%as they do not provide worst case guarantees.

%!TEX root = attributePrivacy.tex
\section{Preliminaries}\label{s.prelim}

\textbf{Pufferfish Privacy.}
The Pufferfish privacy framework \cite{kifer2014pufferfish} consists of three components: a set of secrets $S$, a set of discriminative pairs $\pufQ \subseteq S \times S$, and a class of data distributions $\Theta$. $S$ is a set of possible facts about the database that we might wish to hide. $\pufQ$ is the set of secret pairs $(s_i,s_j)$, $s_i, s_j \in S$, that we wish to be indistinguishable, where $s_i$ and $s_j$ must be mutually exclusive. The class of data distributions $\Theta$ can be viewed as a set of conservative assumptions about the underlying distribution that generates the database.

\begin{definition}[$(\epsilon,\delta)$-Pufferfish Privacy \cite{kifer2014pufferfish,song2017pufferfish}\protect\footnotemark]\label{def:pufferfish} \footnotetext{The original definition~\cite{kifer2014pufferfish}
and the one considered in~\cite{song2017pufferfish}
is $(\epsilon,0)$-Pufferfish. We extend the definition to $(\epsilon,\delta)$-Pufferfish in the natural way.} A mechanism $\mathcal{M}$ is $(\eps,\delta)$-Pufferfish private in a framework $(S,\pufQ,\Theta)$ if for all $\theta\in \Theta$ with $X \sim \theta$, for all secret pairs $(s_i,s_j)\in \pufQ$ such that $P(s_i|\theta)\neq 0$ and $P(s_j|\theta)\neq 0$, and for all $T \subseteq Range(\mathcal{M})$, we have
%	\begin{equation}
\[
P_{\mathcal{M},\theta}(\mathcal{M}(X)\in T|s_i,\theta)\le \exp(\eps) {P_{\mathcal{M},\theta}(\mathcal{M}(X)\in T|s_j,\theta)} + \delta.
\]
%	\end{equation}
%	where $s_i$ and $s_j$ are such that $P(s_i|\theta)\neq 0$, $P(s_j|\theta)\neq 0$.
\end{definition}
%

%\rc{decide to move some of this text to the corresponding sections now that we dont have it in the appendix anymore}
The Wasserstein Mechanism proposed
in~\cite{song2017pufferfish} and defined formally in Section \ref{s.wass} is the first
general mechanism for satisfying instantiations of Pufferfish privacy framework.  
It defines sensitivity of a function $F$ as the maximum Wasserstein distance between the distribution of $F(X)$ given two different realizations of secrets $s_i$ and $s_j$ for $(s_i,s_j) \in \pufQ$.  The mechanism then instantiates the Laplace mechanism by outputting $F(X)$ plus Laplace noise that scales with this sensitivity. 
Although this mechanism works in general for any instantiation of the Pufferfish framework, computing Wasserstein distance for all pairs of secrets is computationally expensive, and will typically not be feasible in practice.  

\cite{song2017pufferfish} also gave the Markov Quilt Mechanism (Algorithm \ref{algo:mq} in Appendix \ref{app.prelim}) for some special structures of data dependence.
It is more efficient than the Wasserstein Mechanism and also guarantees $(\eps,0)$-Pufferfish privacy. 
%Both mechanisms are given formally in the appendix.
%\rc{this doesn't flow, here to end of section}
The Markov Quilt Mechanism  of \cite{song2017pufferfish} assumes that the entries in the input database $Y$ form a Bayesian network, as defined below.  These entires could either be: (1) the multiple attributes of a single record when the database contained only one record, or (2) the attribute values across multiple records for a single-fixed attribute when the database contained multiple attributes.  Hence, the original Markov Quilt Mechanism could not accommodate correlations across multiple attributes in multiple records, as we study in this work. Full details of this algorithm are given in Appendix \ref{app.prelim}.

\begin{definition}[Bayesian Networks]\label{def.bayesnet}
	A Bayesian network is described by a set of variables $Y=\{Y_1,\ldots, Y_n\}$ and a directed acyclic graph $G=(Y,E)$ whose vertices are variables in $Y$. The probabilistic dependence on  $Y$ included by the network can be written as: $\Pr(Y_1,\ldots,Y_n)=\Pi_{i=1}^n \Pr(Y_i|\parent(Y_i))$.
%	\[ \Pr(Y_1,\ldots,Y_n)=\Pi_{i=1}^n \Pr(Y_i|\parent(Y_i)).\] 
\end{definition}

\section{Attribute Privacy Definitions}\label{s.attdefs}

\paragraph{Data model and representation.}
The dataset $X$ contains $n$ records, where each record consists of $m$ attributes. We view the dataset $X$ as an $n\times m$ matrix.
%, where each row corresponds to an individual, and each column corresponds to an attribute.
In this work, we are interested in privacy of the \emph{columns}, which
represent attributes that a data owner wishes to protect.
Thus we refer to the matrix $X$ as $X=[X_1,\ldots,X_m]$, where $X_i$ is the column vector related to the $i$th attribute (column). 
In contrast, traditional differential privacy~\cite{DMNS06,privacybook}
is concerned with privacy of the \emph{rows} of the dataset matrix.
We let $X^j_i$ denote $i$th attribute value for $j$th record.

Each record is assumed to be sampled i.i.d.~from an unknown distribution, where attributes within a single record can be correlated (e.g., consider height and weight).  We use $C \subseteq [m]$ to denote a set of indices of the sensitive attributes that require privacy protection (e.g., race and gender may be sensitive attributes; hair color may be non-sensitive). The data owner wishes to compute a function $F$ over her dataset
and release the value (or estimate of the value) $F(X)$ while protecting some information about the sensitive attributes.

%The $j$th column is considered to be sampled i.i.d from some distributions with parameter $\theta_j$. Let $\theta_S$ denote the parameters of the sensitive attributes.

%The database $X$ consists of several attributes $X_1, X_2, \ldots, X_m$, among which some attributes are sensitive. We have $n$ people in the database. We denote $y_j=(x_{1,j},x_{2,j},\ldots,x_{m,j})$ as the individual $j$'s data. The attributes might be correlated.  Let $X_N$ and $X_S$ denote the non-sensitive attributes and the sensitive attributes, respectively. The goal is to release aggregate statistics $F(X_N)$ of the non-sensitive attributes $X_N$, while ensuring the privacy of the sensitive attributes. Since we wish to prevent the information leakage of the sensitive attributes in the population-level, differential privacy is not applicable. 

\paragraph{Privacy notions.}
We distinguish between three kinds of attribute privacy, corresponding to three different types of information the data owner may wish to protect.

\emph{Individual attribute privacy} protects $X^j_i$ for sensitive attribute~$i$ when $F(X)$ is released. Note that differential privacy provides individual attribute privacy simultaneously for all individuals and all attributes \cite{DMNS06}, but does not protect against individual-level inferences from population-level statistics \cite{10.1145/2020408.2020598}. For example, if a DP result
shows a correlation between lung disease and smoking, one may infer that a known-smoker in the dataset has an elevated likelihood of lung disease.

%wishes to protect $X^j_i$ when $F(X)$ is released
%where $i$th attribute is sensitive (e.g., race of a person).
%Note that differential privacy provides
%individual attribute privacy to the extent of
%what $F$ reveals about the population~\cite{privacybook,10.1145/2020408.2020598} (e.g., if a DP result
%shows that there is correlation between some desease and smoking
%then if an individual participating in a dataset smokes one may infer
%something about the chance of this person having the desease).

\emph{Dataset attribute privacy} is applicable when the owner wishes to reveal $F(X)$
while protecting the value of some function $g(X_i)$ for sensitive attribute $i\in C$ (e.g., whether there were more Caucasians or Asians present in the dataset).

%is considered in the setting
%where the owner wishes to reveal $F(X)$
%while protecting
%some function $g(X_i)$, $i\in C$,
%over values of a sensitive attribute
%(e.g., whether there were more Caucasians or Asians present
%in the dataset).

\emph{Distribution attribute privacy} protects privacy of a parameter $\phi_i$ that
governs the distribution of $i$th sensitive attribute in
the underlying population from which the data are sampled.

%Finally, \emph{distribution attribute privacy}
%considers
%privacy of the parameter $\phi_i$ that
%governs the distribution of $i$th sensitive attribute in
%the underlying population where the dataset
%is sampled from.

The last two notions are the ones put forward in this paper
and studied in detail.
The difference between them may be subtle
depending on $g$ and $\phi$.
For example, consider one setting where the sensitive attribute
is binary and $g$ is the fraction of records
where this attribute is 1, and another setting
where the sensitive attribute is a Bernoulli random
variable with parameter $\phi$.
In this case, $g$ can be seen as an estimate of $\phi$
based on a sample.
The difference becomes particularly relevant in settings where privacy is required for realizations of the dataset that are unlikely under the data distribution, or settings with small datasets where $g$ is a poor estimate of $\phi$.

\paragraph{Formal framework for attribute privacy.}
The standard notion of differential privacy is not
directly applicable to our setting since
we are interested in protecting
population-level information.
Instead, we formalize our attribute privacy definitions using the Pufferfish privacy framework of Definition \ref{def:pufferfish} by specifying the three components $(S,\pufQ,\Theta)$. The distributional assumptions of this framework are additionally useful for formalizing correlation across attributes. 

%This framework also allows us to specify a data distribution over records, which we use to formalize correlation across attributes.

%To this end,
%we instantiate the Pufferfish privacy framework by setting the three components
%$(S,Q,\Theta)$.

\begin{definition}[Dataset Attribute Privacy]
Let $(X_1^j, X_2^j,\ldots,X_m^j)$ be a record with $m$ attributes that is sampled from an unknown distribution $\mathcal{D}$, and let $X=[X_1,\ldots,X_m]$ be a dataset of $n$ records sampled i.i.d.~from $\mathcal{D}$ where $X_i$ denotes the (column) vector containing values of $i$th attribute of every record. Let $C\subseteq [m]$ be the set of indices of sensitive attributes, and for each $i \in C$, let $g_i(X_i)$ be a function with codomain $\mathcal{U}^i$.

A mechanism $\mathcal{M}$ satisfies \emph{$(\epsilon,\delta)$-dataset attribute privacy} if it is $(\epsilon,\delta)$-Pufferfish private for the following framework $(S,\pufQ,\Theta)$:
\begin{description}
\item Set of secrets: $S=\{s^i_a := \mathbbm{1}[g_i(X_i) \in \mathcal{U}^i_a] : \mathcal{U}^i_a \subseteq \mathcal{U}^i, i\in C\}$.
\item Set of secret pairs: $\pufQ=\{(s_a^i,s_b^i) \in S \times S, i\in C\}$.
\item Distribution: $\Theta$ is a set of possible distributions $\theta$ over the dataset $X$. For each possible distribution $\mathcal{D}$ over records, there exists a $\theta_{\mathcal{D}} \in \Theta$ that corresponds to the distribution over $n$ i.i.d.~samples from $\mathcal{D}$.
\end{description}
\label{def:dbprivacy}
\end{definition}

%\rc{describing g, S, Q, theta in text}
This definition defines each secret $s_a^i$ as the event that $g_i(X_i)$ takes a value in a particular set $\mathcal{U}_a^i$, and the set of secrets $S$ is the collection of all such secrets for all sensitive attributes.  This collection may include all possible subsets of $\mathcal{U}^i$, or it may include only application-relevant events.  For example, if all $\mathcal{U}_a^i$ are singletons, this corresponds to protecting any realization of $g_i(X_i)$. Alternatively, the data owner may only wish to protect whether $g_i(X_i)$ is positive or negative, which requires only $\mathcal{U}_a^i = (-\infty,0)$ and $\mathcal{U}_b^i = [0,\infty)$. The set of secret pairs $\pufQ$ that must be protected includes all pairs of the events on the same sensitive attribute. The Pufferfish framework considers distributions $\theta$ over the entire dataset $X$, whereas we require distributions $\mathcal{D}$ over records.  We resolve this by defining $\Theta$ to be the collection of distributions over datasets induced by the allowable i.i.d.~distributions over records.

Determining which functions $g_i$ to consider is an interesting question. For example, in~\cite{10.1145/1142351.1142375}
the authors show that it is tractable to check whether the output of certain classes of functions evaluated on a dataset reveals information about the output of another query evaluated 
on the same dataset. Hence, given a function $F$ whose output a data owner wishes to release,
the owner may consider either those $g_i$'s about which $F$ reveals information, or
those for which verifying perfect privacy w.r.t. $F$ is infeasible.

\begin{definition}[Distributional Attribute Privacy]
Let $(X_1^j, X_2^j,\ldots,X_m^j)$ be a record with $m$ attributes that is sampled from an unknown distribution described by a vector of random variables $(\phi_1,\ldots,\phi_m)$, where $\phi_i$
parameterizes the marginal distribution of $X_i^j$ conditioned on the values of all $\phi_k$ for $k\neq i$. The $(\phi_1,\ldots,\phi_m)$ are drawn from a known joint distribution $P$, and each $\phi_i$ has support $\Phi^i$.
Let $X=[X_1,\ldots,X_m]$ be a dataset of $n$ records sampled i.i.d.~from the distribution described by $(\phi_1,\ldots,\phi_m)$
where $X_i$ denotes the (column) vector containing values of $i$th attribute of every record. Let $C\subseteq [m]$ be the set of indices of sensitive attributes. 

%Let $(X_1^j, X_2^j,\ldots,X_m^j)$ be a record with $m$ attributes, and let $X=[X_1,\ldots,X_m]$ be a database of $n$ records sampled i.i.d.~from unknown distribution $\theta$
%where $X_i$ denotes the (column) vector containing values of $i$th attribute of every record. The distribution $\theta$ 

% sampled according to $\theta$ where
%$X_i^j$ follows a known distribution with parameter
%$\phi_i\in{\Phi}$ \rc{phi or U?} that can be a random variable.
%let $X=[X_1,\ldots,X_m]$ be a database of $n$ records sampled i.i.d.~from $\theta$
%where $X_i$ denotes the (column) vector containing values of $i$th attribute of every record, and let $C\subseteq [m]$ be the set of indices of sensitive attributes. 

A mechanism $\mathcal{M}$ satisfies \emph{$(\epsilon,\delta)$-distributional attribute privacy} if it is $(\epsilon,\delta)$-Pufferfish private for the following framework $(S,\pufQ,\Theta)$:
\begin{description}
\item Set of secrets: $S=\{s^i_a := \mathbbm{1}[\phi_i \in \Phi^i_a] : \Phi^i_a \subset \Phi^i, i\in C\}$.
\item Set of secret pairs: $\pufQ=\{(s_a^i,s_b^i) \in S \times S, i\in C\}$.
\item Distribution: $\Theta$ is a set of possible distributions $\theta$ over the dataset $X$. For each possible $\phi = (\phi_1,\ldots,\phi_m)$
describing the conditional marginal distributions for all attributes, there exists a $\theta_{\phi} \in \Theta$ that corresponds to the distribution over $n$ i.i.d.~samples from the distribution over records described by $\phi$.
\end{description}
\label{def:distprivacy}
\end{definition}

%\rc{para describing distributional stuff}

This definition naturally parallels Definition \ref{def:dbprivacy}, with the attribute-specific random variable $\phi_i$ taking the place of the attribute-specific function $g_i(X_i)$.  Although it might seem natural for $\phi_i$ to define the \emph{marginal} distribution of the $i$th attribute, this would not capture the correlation across attributes that we wish to study.  Instead, $\phi_i$ defines the \emph{conditional marginal} distribution of the $i$th attribute given all other $\phi_{\neq i}$, which does capture such correlation. This also allows the distribution $\theta$ over datasets to be fully specified given these parameters and the size of the dataset.

More specifically, we model attribute distributions using standard notion of Bayesian hierarchical modeling. The $(\phi_1,\ldots,\phi_m)$ can be viewed as a set of hyperparameters of the distributions of the attributes, and $P$ as hyper-priors of the hyperparameters.
The distribution $P$ is captured in~$\Theta$, and the distribution of attribute $X_i$ is governed by a realization of the random variable $\phi_i$. The $\phi_i$ describes the conditional marginal distribution for attribute~$i$: it is the hyperparameter of the probability of $X_i$
given hyperparameters of all other attributes $P(X_i|\phi_1,\ldots, \phi_{i-1}, \phi_{i+1},\ldots, \phi_m)$.
We make the ``naive'' conditional independence assumption that all attributes $X_i$ are mutually independent conditional on the set of parameters $(\phi_1,\ldots,\phi_m)$, hence, $(\phi_1,\ldots,\phi_m)$ fully capture the distribution of a record. The ``naive'' conditional independence is a common assumption in probabilistic models, and naive Bayes is a simple example that employs this assumption.

Since both of our attribute privacy definitions are instantiations of the Pufferfish privacy framework, one could easily apply the Wasserstein Mechanism \cite{song2017pufferfish} to satisfy $(\epsilon,0)$-attribute privacy for either of our definitions.  {The Wasserstein distance metric has also been used to calibrate noise in prior work on distributional variants of differential (individual-level) privacy~\cite{KM19a,KM19b}.} However, as described in Section \ref{s.prelim}, implementing this mechanism requires computing Wasserstein distance between the conditional distribution on $F(X)$ for all pairs of secrets in $\pufQ$. Computing exact Wasserstein distance is known to be computationally expensive, and our settings may require exponentially many computations in the worst case. In the remainder of the paper, we provide efficient algorithms that satisfy each of these privacy definitions, focusing on dataset attribute privacy in Section \ref{sec:datasecrets} and distributional attribute privacy in Section \ref{sec:distsecrets}, before returning to the (inefficient) Wasserstein Mechanism in Section \ref{s.wass}.

\section{The Gaussian Mechanism for Dataset Attribute Privacy}\label{sec:datasecrets}

In this section we consider \dbprivacy{} as introduced in Definition \ref{def:dbprivacy}. In this setting, an analyst wants to publish a function $F$ evaluated on her dataset $X$, but is concerned about an adversary observing $F(X)$ and performing a Bayesian update to make inferences about a protected quantity $g_i(X_i)$. We propose a variant of the Gaussian Mechanism \cite{privacybook} that satisfies \dbprivacy{} when $F(X)$ conditioned on $g_i(X_i)$ follows a Gaussian distribution, with constant variance conditioned on $g_i(X_i)=a$ for all $a$. Although this setting is more restrictive, it is still of practical interest.
For example, it can be applied when $X$ follows a multivariate Gaussian distribution and
$g_i$ and $F$ are linear with respect to the entries of~$X$, as we show in the instantiation
of our mechanism in Section~\ref{sec:gausexample}.
We also note that using
variational auto-encoders (VAEs)~\cite{DBLP:journals/corr/KingmaW13,Higgins2017betaVAELB}, it is possible
to encode
data from other distributions using a Gaussian representation with interpretable features.
This would then allow an analyst to specify which latent features are deemed sensitive
for the data, even if the original features are less descriptive (e.g., pixels on an image vs.~the gender of
the person in it). In Section \ref{s.nogauss}, we propose
the Attribute-Private Gaussian Mechanism for non-Gaussian data that does not make the above assumptions.
In particular, the mechanism allows the analyst to use Gaussian approximations to characterize the conditional distribution of $F(X)$ given $g_i(X_i)$,
while still providing formal dataset attribute privacy guarantees.

%and propose a mechanism
%for settings where the the query answer $F(X)$ conditional on $g(X_i)$, $i\in C$ is Gaussian distributed.
%
%%\subsection{The setting}\label{data:setting}
%When the secrets are some properties of the record in the database, we model the secret as $g(X_i)$, where $g$ is a vector-based real-valued function, here $X_i$ is the column vector.  We consider a more restricted setting when the distribution of the query answer $F(X)$ conditional on $g(X_i)$, $i\in C$ is Gaussian distributed. This setting is still of practical interest, as it is motived by VAE stuff. \todo{add VAE motivation}.

%\paragraph{The framework.} Let $s_a^i$ denote the event that $g(X_i)$ takes value  in the set $\mathcal{U}_a$. The set of secrets is $S=\{s^i_a: \mathcal{U}_a \subset \mathcal{U}, i\in C\}$. The set of secret pairs is $Q=\{(s_a^i,s_b^i): \mathcal{U}_a,\mathcal{U}_b \subset \mathcal{U}, \mathcal{U}_a\cap \mathcal{U}_b=\emptyset , i\in C\}$. $\Theta$ is a set of possible underlying distributions that generate the data such that $F(X)|g(X_i)$, $i\in C$ is Gaussian distributed. Each $\theta\in \Theta$ represents a belief an aversary may hold about the data. \todo{rephrase, this is mostly just copied from the previous section}

\subsection{Attribute-Private Gaussian Mechanism}
\label{sec:gaus_precise}
Algorithm \ref{algo:gaussian} presents the Attribute-Private Gaussian Mechanism for answering a real-valued query $F(X)$ while protecting the values of $g_i(X_i)$ for $i \in C$. Much like the Gaussian Mechanism for differential privacy \cite{privacybook}, the  Attribute-Private Gaussian Mechanism first computes the true value $F(X)$, and then adds a Gaussian noise term with mean zero and standard deviation that scales with the sensitivity of the function.  However, \emph{sensitivity} of $F$ in the attribute privacy setting is defined with respect to each secret attribute $X_i$ as, 
\begin{equation}\label{def:gaussian_sens}
\Delta_i F=\max_{\theta\in \Theta}\max_{(s_a^i,s_b^i)\in \pufQ}\abs{\E{F(X)|s_a^i,\theta}-\E{F(X)|s_b^i,\theta}}.
\end{equation}
This differs from the sensitivity notion used in differential privacy in two key ways. First, we are concerned with measuring changes to the value of $F(X)$ caused by changing secrets $s_a^i$ corresponding to realizations of $g_i(X_i)$, rather than by changing an individual's data.  Second, we assume our data are drawn from an unknown underlying distribution $\theta$, so $F(X)$ is a random variable. Our attribute privacy sensitivity bounds the maximum change in posterior expected value of $F(X)$ in the worst case over all distributions and pairs of secrets for each attribute. We note that if $F(X)$ is independent of the protected attribute $X_i$, then $\Delta_i F = 0$ and no additional noise is needed for privacy.  The Attribute-Private Gaussian Mechanism of Algorithm \ref{algo:gaussian} further benefits from the inherent randomness of the output $F(X)$. In particular, it reduces the variance $\sigma^2$ of the noise added by the conditional variance of $F(X)$ given $g_i(X_i)$ and $\theta$, as the sampling noise can mask some of the correlation. Hence, privacy also comes for free if the function of interest has low correlation with the protected attributes.

Algorithm \ref{algo:gaussian} can be easily extended to handle vector-valued queries with $F(X) \in \mathbb{R}^k$ and sensitive functions $g_i$ over multiple attributes by changing $\Delta_i F$ in Equation \eqref{def:gaussian_sens} to be the maximum $\ell_2$ distance rather than absolute value.  Additionally, the noise adjustment for each attribute should be based on the conditional covariance matrix of $F(X)$ rather than the conditional variance.

%Our Gaussian Mechanism can be generalized for vector-output query mapping to $\mathcal{R}^k$ and secrets function $g$ on multiple attributes.
%Specific changes will be to the sensitivity definition in (\ref{def:gaussian_sens}) --- it will be using the $l_2$ distance rather than the absolute value,
%and to the noise --- noice for each entry will be adjusted with the covariance matrix rather than a single value's variance.
%  \wz{I prefer not to do the generalized version. The covariance matrix is so messy and ugly.}

%{\centering
%	\begin{minipage}{\linewidth}
		\begin{algorithm}[tbh]
		\centering
			\caption{Attribute-Private Gaussian Mechanism, APGM($X, F, \{g_i\}, C, \{S, \pufQ, \Theta \}, \epsilon, \delta$) for dataset attribute privacy.} 
			\begin{algorithmic}
				\State \textbf{Input:} dataset $X$, query $F$, functions $g_i$ for protected attributes $i\in C$, framework $\{S, \pufQ, \Theta \}$ , privacy parameters $\epsilon$, $\delta$
				\State Set $\sigma^2=0$, $c=\sqrt{2\log (1.25/\delta)}$.
				\State \textbf{for} each $i \in C$ \textbf{do}
				\State \indent Set $\Delta_i F=\max_{\theta\in \Theta}\max_{(s_a^i,s_b^i)\in \pufQ}\abs{\E{F(X)|s_i^a}-\E{F(X)|s_i^b}}$.
				%\State \fixme{If we like it:  Set $\mathsf{min\_var} = \min_{\theta \in \Theta}\Var(F(X)|g_i(X_i),\theta)$}
				\State \indent \textbf{if} $(c\Delta_i F/\epsilon)^2-\min_{\theta \in \Theta}\Var(F(X)|g_i(X_i),\theta)\ge \sigma^2$ \textbf{then}
				\State \indent \indent Set $\sigma^2=(c\Delta_i F/\epsilon)^2-\min_{\theta \in \Theta}\Var(F(X)|g_i(X_i),\theta)$.
				\State \textbf{if} $\sigma^2 > 0$ \textbf{then}
				\State \indent Sample $Z\sim\mathcal{N}(0, \sigma^2)$.
				\State \indent Return $F(X)+Z$.
				\State \textbf{else} Return $F(X)$.
			\end{algorithmic}\label{algo:gaussian}
		\end{algorithm}
%	\end{minipage}
%}

\begin{restatable}{theorem}{gausspriv}\label{thm.gausspriv}
%\begin{theorem}[Privacy]\label{thm.gausspriv}
	The Attribute-Private Gaussian Mechanism \\ APGM($X, F, \{g_i\}, C, \{S, \pufQ, \Theta \}, \epsilon, \delta$) in Algorithm \ref{algo:gaussian}  is $(\eps,\delta)$-dataset attribute private when $F(X)|g_i(X_i)$ is Gaussian distributed for any $\theta \in \Theta$ and $i\in C$.
%\end{theorem}
\end{restatable}

%\todo{add short description of the proof strategy and highlight why its not just proof of Gaussian mech DP.}  

Privacy follows from the observation that the summation of $F(X)$ and the Gaussian noise $Z$ is Gaussian distributed conditioned on any secrets, and the probabilities of the output conditioned on any pairs of secrets have the same variance with mean difference $\Delta_i F$. Since we bound the ratio of the two probabilities caused by shifting this variable, the analysis reduces to the proof of Gaussian mechanism in differential privacy. 

\begin{proof}
	Fix any pair of secrets $(s_a^i,s_b^i) \in \pufQ$ for a fixed secret attribute $X_i$ under any $\theta \in \Theta$. Recall that $s_a^i$ denotes the event that $g_i(X_i)\in\mathcal{U}^i_a$. %\todo{It's supposed to be $g_i(X_i)\in\mathcal{U}^i_a$. Can we keep this proof for single value $g_i(X_i)=a$ and have short technical argument that it extends to more general sets?} \wz{it doesnt matter in the proof so i just changed it to be the general set.} 
	Let $Z\sim \mathcal{N}(0,\sigma^2)$ denote the Gaussian noise added in Algorithm \ref{algo:gaussian}.  We have $[\mathcal{M}(X)|s_a^i,\theta ]=[(F(X)+Z)|s_a^i,\theta ]=[F(X)|s_a^i,\theta]+Z$, because $Z$ is independent of $s_a^i$ and $\theta$. Since we have assumed that $F(X)|g_i(X_i)$ is Gaussian distributed and the summation of two Gaussians is Gaussian, $\mathcal{M}(X)|s_a^i$ follows a Gaussian distribution with mean $\E{F(X)|s^i_a,\theta}$ and variance $\Var(F(X)|g_i(X_i), \theta)+\sigma^2$. 
	%We will bound from above and below the ratio of the probabilities that the algorithm outputs any $w$ on the pair of secrets $(s_a^i,s_b^i)$ as follows: 
	The ratio of probabilities of seeing an output $w$ on a pair of secrets $(s_a^i,s_b^i)$ in the worst case is as follows,
	\begin{align}
	&\max_{(s_a^i,s_b^i)\in \pufQ}\abs{\log\frac{\exp(-\frac{1}{2}(\Var(F(X)|g_i(X_i),\theta)+\sigma^2)(w-\E{F(X)|s_i^a,\theta})^2}{\exp(-\frac{1}{2}(\Var(F(X)|g_i(X_i),\theta)+\sigma^2)(w-\E{F(X)|s_i^b,\theta})^2}} \notag\\
	=&\abs{\log\frac{\exp(-\frac{1}{2}(\Var(F(X)|g_i(X_i),\theta)+\sigma^2)w^2)}{\exp(-\frac{1}{2}(\Var(F(X)|g_i(X_i),\theta)+\sigma^2)(w+\Delta)}},\label{eq.bound_1}
	\end{align}
	where $\Delta=\max_{(s_i^a,s_i^b)\in \pufQ}\E{F(X)|s^i_a,\theta}-\E{F(X)|s^i_b,\theta}$.
	Equation \eqref{eq.bound_1} follows from shifting the variable $w$ to $w+\E{F(X)|s^i_a,\theta}$. We observe that the probability ratio can be viewed as the probability ratio in the Gaussian Mechanism in {\em differential privacy} with noise draw from $\mathcal{N}(0,\Var(F(X)|g_i(X_i),\theta)+\sigma^2)$, and query sensitivity $\max_{(s_a^i,s_b^i)\in \pufQ}\E{F(X)|s_a^i}-\E{F(X)|s_b^i}$.  Then, our analysis reduces to the proof of the Gaussian Mechanism in {\em differential privacy}. The Gaussian Mechanism for {\em differential privacy} with 
		\begin{align*}
		&\Var(F(X)|g_i(X_i),\theta)+\sigma^2\\\ge &2\log (1.25/\delta)\left(\frac{\max_{(s_a^i,s_b^i)\in \pufQ}\E{F(X)|s_a^i,\theta}-\E{F(X)|s_b^i,\theta}}{\eps}\right)^2
		\end{align*}ensures that with probability at least $1-\delta$, we have 
	\begin{equation*}
	P(\mathcal{M}(X)\in T |s_a^i,\theta)\le \exp(\eps)P(\mathcal{M}(X)\in T |s_b^i,\theta)+\delta,
	\end{equation*}
	for any $T \subseteq Range(\mathcal{M})$.
%	Thus, setting
%	\begin{equation*}
%	\Var(F(X)|g_i(X_i),\theta)+\sigma^2\ge 2\log (1.25/\delta)\left(\frac{\max_{(s_a^i,s_b^i)\in \pufQ}\E{F(X)|s_a^i,\theta}-\E{F(X)|s_b^i,\theta}}{\eps}\right)^2
%	\end{equation*}
%	suffices to bound \eqref{eq.bound_1} to be less than $\eps$ with probability at least $1-\delta$.
	 Equivalently, we have 
	\begin{align*}
	\sigma^2\ge &2\log (1.25/\delta)\left(\frac{\max_{(s_a^i,s_b^i)\in \pufQ}\E{F(X)|s_a^i,\theta}-\E{F(X)|s_b^i,\theta}}{\eps}\right)^2\\&-\Var(F(X)|g_i(X_i),\theta). 
	\end{align*}
	Taking the maximum over $i$ for all secret attributes and over all $\theta \in \Theta$, we require 
	\begin{equation*}
	\sigma^2\ge \max_{i\in C}\left[  2\log (1.25/\delta)(\Delta_i F/\eps)^2- \min_{\theta \in \Theta}\Var(F(X)|g_i(X_i),\theta)\right] 
	\end{equation*}
	which will ensure that the ratio of the probabilities that the algorithm $\mathcal{M}(X)$ outputs a query value in any subset $T$ on any pair of secrets for any $\theta \in \Theta$ is bounded by $\eps$ with probability at least $1-\delta$.
\end{proof}

High probability additive accuracy bounds on the output of Algorithm \ref{algo:gaussian} can be derived using tail bounds on the noise term $Z$ based on its variance $\sigma^2$. The formal accuracy guarantee is stated in Theorem~\ref{thm.gaussacc}, which follows immediately from tail bounds of a Gaussian.

\begin{restatable}{theorem}{gaussacc}\label{thm.gaussacc}
	%\begin{theorem}[Privacy]\label{thm.gausspriv}
	The Attribute-Private Gaussian Mechanism \\ APGM($X, F, \{g_i\}, C, \{S, \pufQ, \Theta \}, \epsilon, \delta$) in Algorithm \ref{algo:gaussian}  is $(\alpha,\beta)$-accurate for any $\beta>0$ and 
	$$\alpha=\sqrt{\max\{0,\max_{i\in C}\{(c\Delta_i F/\epsilon)^2-\min_{\theta \in \Theta}\Var(F(X)|g_i(X_i),\theta)\}\}}\Phi^{-1}(1-\frac{\beta}{2}),$$
	where $c=\sqrt{2\log (1.25/\delta)}$ and $\Phi$ is the CDF of the standard normal distribution.
	%\end{theorem} 
\end{restatable}

 In general, if $F(X)$ is independent of, or only weakly correlated with the protected functions $g_i(X_i)$, then no noise is needed is preserve dataset attribute privacy, and the mechanism can output the exact answer $F(X)$. On the other hand, if $F(X)$ is highly correlated with $g_i(X_i)$, we then consider a tradeoff between the sensitivity and the variance of $F(X)$. If the variance of $F(X)$ is relatively large, then $F(X)$ is inherently private, and less noise is required. If the variance of $F(X)$ is small and the sensitivity of $F(X)$ is large,
the mechanism must add a noise term with large $\sigma^2$, resulting in low accuracy with respect to the true answer.  To make these statements more concrete and understandable, Section \ref{sec:gausexample} provides a concrete instantiation of Algorithm \ref{algo:gaussian}.

\subsection{Privacy guarantees without Gaussian assumptions}\label{s.nogauss}

A natural question is how we can apply the Attribute-private Gaussian Mechanism when the distributional assumptions of Section \ref{sec:gaus_precise} do not hold. That is, when the actual distribution of $F(X)$ conditioned on $g_i(X_i)$ is not Gaussian distributed.
Our idea is based on using a collection of Gaussian distributions to approximate the actual distributions.
In this section, we show that the Attribute-private Gaussian Mechanism can still be applied in this case to achieve distributional attribute privacy.

We quantify this distributional closeness using $\eta$-approximate max-divergence. The max-divergence and approximate max-divergence are defined as follows.
\begin{definition}[Max-Divergence]
	Let $p$ and $q$ be two distributions with the same support. The max-divergence $D(p||q)$ between them is defined as: 
	$$D(p||q)=\sup_{T \subset \text{Support}(p)} \log \frac{\Pr(p(x)\in T)}{\Pr(q(x)\in T)}.$$
\end{definition}

\begin{definition}[$\eta$-Approximate Max-Divergence]
	Let $p$ and $q$ be two distributions. The $\eta$-approximate\footnote{The approximation parameter is typically named $\delta$ in the literature; we use $\eta$ instead to avoid confusion with the privacy parameter.} max-divergence between them is defined as:
	$$D^\eta(p||q)=\sup_{T \subset \text{Support}(p): \Pr[p(x)\in T]\ge \eta} \log \frac{\Pr(p(x)\in T)-\eta}{\Pr(q(x)\in T)}.$$
\end{definition}

We will consider the following variant of $\eta$-approximate max-divergence:
\begin{equation}\label{appro_var}
\mathcal{D}^\eta(f_{s_a^i,\theta}, \tilde{f}_{s_a^i,\theta}):=\max\{D^{\eta}(f_{s_a^i,\theta}||\tilde{f}_{s_a^i,\theta}), D^{\eta}(\tilde{f}_{s_a^i,\theta}||f_{s_a^i,\theta})\}.
\end{equation}

%We will use $\eta$ as the approximation parameter to distinguish it from the privacy parameter $\delta$.\olya{why not use $\eta$ to begin with in Def 7?}
Formally, let $f_{s_a^i,\theta}$ denote the actual conditional distributon of $F(X)$ given the secret $s_a^i \in S$ and $\theta  \in \Theta$. 
For each $f_{s_a^i,\theta}$, let $\tilde{f}_{s_a^i,\theta}$ denote a Gaussian appromation to $f_{s_a^i,\theta}$. 
For any $\eta>0$, let $\lambda_{\eta}$ be the bound such that for every $s_a^i$, the (variant) approximate max-divergence $\mathcal{D}^{\eta}(f_{s_a^i,\theta}, \tilde{f}_{s_a^i,\theta}) \le \lambda_{\eta}$. That is, $\lambda_{\eta}$ is a constant determined by $\eta$.
%The (variant) approximate max-divergence $\mathcal{D}^{\eta}(f_{s_a^i,\theta}, \tilde{f}_{s_a^i,\theta})$ is bounded above by $\lambda_{\eta}$  for any $\eta>0$, where $\lambda_{\eta}$ is a constant determined by $\eta$.  \rc{do we assume this, or require it, or are we just defining $\lambda_{\eta}$?}\wz{It's just a constant wrt $\eta$, which serves as the upper bound. It appears in the privacy guarantees. So I'd say the Gaussian approximations are chosen to be close to the actual distributions so that $\lambda_{\eta}$ is small.}
%\olya{How about:
%Let $\lambda_{\eta}$ be such that for every $s_a^i$ the (variant) approximate max-divergence $\mathcal{D}^{\eta}(f_{s_a^i,\theta}, \tilde{f}_{s_a^i,\theta}) \le \lambda_{\eta}$  for any $\eta>0$. }
For any fixed $\theta \in \Theta$ and $i\in C$, the collection of $\{\tilde{f}_{s_a^i,\theta}\}$ is chosen to have constant variance for any $s_a^i\in S$. We denote this variance as $\Var (\tilde{f}_{i,\theta})$. The Attribute-Private Gaussian Mechnism for non-Gaussian data of Algorithm \ref{algo:gaussian2} allows the analyst to choose a set of Gaussian approximations $\{\tilde{f}_{s_a^i,\theta}\}$ for each secret $s_a^i \in S$ and $\theta  \in \Theta$, and use it instead of the actual conditional distribution of $F(X)$ in the rest of the algorithmic steps, which are the same as in the Attribute-Private Gaussian Mechanism of Algorithm~\ref{algo:gaussian}.

\begin{algorithm}[tbh]
	\centering
	\caption{Attribute-Private Gaussian Mechanism for non-Gaussian data, APGMnG($X, F, \{g_i\}, C, \{S, \pufQ, \Theta \}, \epsilon, \delta$) .} 
	\begin{algorithmic}
		\State \textbf{Input:} dataset $X$, query $F$, functions $g_i$ for protected attributes $i\in C$, framework $\{S, \pufQ, \Theta \}$ , privacy parameters $\epsilon$, $\delta$
		\State Set $\sigma^2=0$, $c=\sqrt{2\log (1.25/\delta)}$.
		\State \textbf{for} each $i \in C$ \textbf{do}
		\State \indent Choose a Gaussian approximation $\tilde{f}_{s_a^i,\theta}$ for each secret \\ \hspace{0.55cm} $s_a^i \in S$ and $\theta  \in \Theta$. %\wz{put this line in the input or here?}\rc{I think it's reasonable here, but I don't feel strongly either way}
		\State \indent Set $\Delta_i F=\max_{\theta\in \Theta}\max_{(s_a^i,s_b^i)\in \pufQ}\abs{\E{\tilde{f}_{s_a^i,\theta}}-\E{\tilde{f}_{s_b^i,\theta}}}$.
		%\State \fixme{If we like it:  Set $\mathsf{min\_var} = \min_{\theta \in \Theta}\Var(F(X)|g_i(X_i),\theta)$}
		\State \indent \textbf{if} $(c\Delta_i F/\epsilon)^2-\min_{\theta \in \Theta}\Var (\tilde{f}_{i,\theta})\ge \sigma^2$ \textbf{then}
		\State \indent \indent Set $\sigma^2=(c\Delta_i F/\epsilon)^2-\min_{\theta \in \Theta}\Var (\tilde{f}_{i,\theta})$.
		\State \textbf{if} $\sigma^2 > 0$ \textbf{then}
		\State \indent Sample $Z\sim\mathcal{N}(0, \sigma^2)$.
		\State \indent Return $F(X)+Z$.
		\State \textbf{else} Return $F(X)$.
	\end{algorithmic}\label{algo:gaussian2}
\end{algorithm}

The following theorem states that if the Gaussian approximations in Algorithm \ref{algo:gaussian2} are close to the actual distributions of $F(X)$ conditioned on $g_i(X_i)$, then the loss in privacy is not too large. We note that Theorem~\ref{thm.gausspriv2}  involves the variant definition of $\eta$-approximate max-divergence introduced in Definition \ref{appro_var}, which does not require that the actual distribution has the same support as the Gaussian distribution; thus, it is also applicable for all discrete distributions or finite-support distributions. 
%\olya{as my comment before can we just not use $\delta$ to refer to approx?}\wz{I've changed it}

\begin{restatable}{theorem}{gausspriv2}\label{thm.gausspriv2}
	The Attribute-Private Gaussian Mechanism for non-Gaussian data APGMnG($X, F, \{g_i\}, C, \{S, \pufQ, \Theta \}, \epsilon, \delta$) in Algorithm \ref{algo:gaussian2}  is $(\eps+2\lambda_{\eta}, \exp(\lambda_{\eta}) \delta+\eta)$-dataset attribute private, when the Gaussian approximation $\tilde{f}_{s_a^i,\theta}$ satisfies $\mathcal{D}^{\eta}(f_{s_a^i,\theta}, \tilde{f}_{s_a^i,\theta}) \leq \lambda_{\eta}$  for any $\eta>0$ for all   $s_a^i \in S$ and $\theta  \in \Theta$.
	%\end{theorem} 
\end{restatable}

%\olya{do we need word "satisfies" or can be removed}\rc{changed it}

\begin{proof}

	We first analyze the probability of the output conditioned on the event that \\ $\{\max_{i\in C}\max_{s_a^i\in S}\max\{D(f_{s_a^i,\theta}||\tilde{f}_{s_a^i,\theta}), D(\tilde{f}_{s_a^i,\theta}||f_{s_a^i,\theta})\}\le \lambda_{\eta}\}$.
    Fix any pair of secrets $(s_a^i,s_b^i) \in \pufQ$ for a fixed sensitive attribute $X_i$ under any $\theta \in \Theta$. 
    Let $Z\sim \mathcal{N}(0,\sigma^2)$ denote the Gaussian noise added in Algorithm \ref{algo:gaussian2}. Let us partition $\mathbb{R}$ as $\mathbb{R}=R_1\cup R_2$, where $$R_1=\{F(X)+Z\in \mathbb{R}: |F(X)+Z|\le c\Delta_iF/\eps\},$$ and $$R_2=\{F(X)+Z\in \mathbb{R}:|F(X)+Z| > c\Delta_iF/\eps\}.$$ Fix any subset $T\subseteq \mathbb{R}$, and define 
    $T_1=T\cap R_1$ and $T_2=T\cap R_2.$
    %where the probability is taking over the randomness of both the mechanism and the data, because the generation process for $z$ is indepedent of the given data.
    
    For any $w\in T_1$, we can write the probability ratio of seeing the output $w$ as follows:
    \begin{align}
    &\frac{\Pr(F(X)+Z=w|s_a^i,\theta)}{\Pr(F(X)+Z=w|s_b^i,\theta)}\notag\\
    =&\frac{\Pr(F(X)+Z=w|F \sim f_{s_a^i,\theta})}{\Pr(F(X)+Z=w|F \sim \tilde{f}_{s_a^i,\theta})}\cdot \frac{\Pr(F(X)+Z=w|F\sim \tilde{f}_{s_b^i,\theta})}{\Pr(F(X)+Z=w|F\sim f_{s_b^i,\theta})}\notag\\
    &\cdot \frac{\Pr(F(X)+Z=w|F\sim \tilde{f}_{s_a^i,\theta})}{\Pr(F(X)+Z=w|F\sim \tilde{f}_{s_b^i,\theta})}. \label{eq.dev}
    \end{align}
    For any $w\in T_1$, the Attribute-Private Gaussian Mechanism for non-Gaussian data ensures that the last ratio in Equation~\eqref{eq.dev} is bounded above by $\exp(\eps)$. For the first ratio in Equation~\eqref{eq.dev}, since the generation process for $z$ is indepedent of the data, we have 
    \begin{align*}
    &\frac{\Pr(F(X)+Z=w|F \sim f_{s_a^i,\theta})}{\Pr(F(X)+Z=w|F \sim \tilde{f}_{s_a^i,\theta})}\\
    =&\frac{\int_f \Pr(Z=w-f)\Pr(F(X)=f|F \sim f_{s_a^i,\theta})df}{\int_f \Pr(Z=w-f)\Pr(F(X)=f|F \sim \tilde{f}_{s_a^i,\theta})df}\\
    \le&\max_f \frac{\Pr(F(X)=f|F \sim f_{s_a^i,\theta})}{\Pr(F(X)=f|F \sim \tilde{f}_{s_a^i,\theta})}\\
    \le&\exp(\lambda_{\eta})
    \end{align*}
    Similarly, we can bound the second ratio by $\exp(\lambda_{\eta})$. Thus, Equation~\eqref{eq.dev} is bounded by $\exp(\eps+2\lambda_{\eta})$, which is equivalent to 
    \begin{equation}\label{eq.T1}
    \Pr(F(X)+Z\in T_1| s_a^i,\theta)\le \exp(\eps+2\lambda_{\eta}) \Pr(F(X)+Z\in T_1| s_b^i,\theta)
    \end{equation}
    We also bound the probability that the output belongs to the subset $T_2$ as follows:
    \begin{align}
    \Pr(F(X)+Z\in T_2| s_a^i,\theta)\le &\exp(\lambda_{\eta}) \Pr(F(X)+Z\in T_2| F \sim \tilde{f}_{s_a^i,\theta})\notag\\
    \le &\exp(\lambda_{\eta}) \delta. \label{eq.T2}
    \end{align}
    Then, by (\ref{eq.T1}) and (\ref{eq.T2}), we have 
\begin{align}
&\Pr(F(X)+Z\in T|s_a^i,\theta)\notag\\=&\Pr(F(X)+Z\in T_1|s_a^i,\theta)+\Pr(F(X)+Z\in T_2|s_a^i,\theta)\notag\\
\le &  \exp(\eps+2\lambda_{\eta}) \Pr(F(X)+Z\in T_1| s_b^i,\theta)+\exp(\lambda_{\eta}) \delta. \label{eq.good}
\end{align}   
 
 We then analyze the probability of the output when \\ $\max_{i\in C}\max_{s_a^i\in S}\max\{D(f_{s_a^i,\theta}||\tilde{f}_{s_a^i,\theta}), D(\tilde{f}_{s_a^i,\theta}||f_{s_a^i,\theta})\}$ is bounded by $\lambda_{\eta}$ with probability at least $1-\eta$, which is equivalent to the $\eta$-approximate max-divergence $\mathcal{D}^{\eta}(f_{s_a^i,\theta}, \tilde{f}_{s_a^i,\theta})$ is bounded above by $\lambda_{\eta}$.
 For any $\eta>\bar{\delta}$, define this high probability event as follows:
	$$E_{\eta}:=\{\max_{i\in C}\max_{s_a^i\in S}\max\{D(f_{s_a^i,\theta}||\tilde{f}_{s_a^i,\theta}), D(\tilde{f}_{s_a^i,\theta}||f_{s_a^i,\theta})\}\le \lambda_{\eta}\}.$$
	Let $E^c_{\eta}$ denote the complement set. By the choice of $\lambda_{\eta}$, we have $\Pr[E^c_{\eta}]\le \eta$. Then by (\ref{eq.good}) and the observation that $\Pr[E^c_{\eta}]\le \eta$, we have that for any subset $T$,
	\begin{align*}
	&\Pr[F(X)+Z\in T|s_a^i,\theta]\\\le &\Pr[F(X)+Z\in T|s_a^i,\theta, E_{\eta}]\Pr[E_{\eta}]+\Pr[E^c_{\eta}]\\
	\le &(\exp(\eps+2\lambda_{\eta})\Pr[F(X)+Z\in T|s_b^i,\theta, E_{\eta}]+\exp(\lambda_{\eta}) \delta)\Pr[E_{\eta}]\\&+\Pr[E^c_{\eta}]\\
	= &\exp(\eps+2\lambda_{\eta})\Pr[F(X)+Z\in T|s_b^i,\theta, E_{\eta}]\Pr[E_{\eta}]\\&+\exp(\lambda_{\eta}) \delta\Pr[E_{\eta}]+\Pr[E^c_{\eta}]\\
	\le &\exp(\eps+2\lambda_{\eta})Pr[F(X)+Z\in T|s_b^i,\theta]+\exp(\lambda_{\eta}) \delta+\eta,
	\end{align*}
    completing the proof.
\end{proof}

\subsection{Instantiation with Gaussian distributed data}
\label{sec:gausexample}
In this section, we show an instantiation of our Attribute-Private Gaussian Mechanism when the joint distribution of the $m$ attributes is multivariate Gaussian. The privacy guarantee of this mechanism requires that $F(X)|g_i(X_i)$ is Gaussian distributed, which is satisfied when $g_i$ and $F$ are linear with respect to the entries of $X$. For simplicity of illustration, we will choose both $F(X)$ and all $g_i(X_i)$ to compute averages.

%mechanism with Gaussian distributed data.
%Recall that a database consists of $n$ individuals (or records) each with $m$ attributes. 
%We consider the case when the joint distribution of the $m$ attributes is multivariate Gaussian. The Gaussian Mechanism for attribute privacy requires that $F(X)|g_i(X_i)$ is Gaussian distributed. It is satisfied when $g$ and $F$ are linear with respect to the value of entries for this case. For simplicity, we take both $g(X_i)$ and $F(X)$ as the average function to demonstrate the instantiation.

%\textbf{Example:}  C
As a motivating example, consider a dataset that consists of students' SAT scores $X_\sat$, heights $X_\hgt$, weights $X_\wgt$, and their family income $X_\inc$.  As a part of a wellness initiative, the school wishes to release the average weight of its students, so $F(X)=\frac1n \sum_{j=1}^n X_\wgt^j$.  The school also wants to prevent an adversary from inferring the average SAT scores or family income of their students, so $C=\{\sat,\inc\}$ and $g(X_i)=\frac1n \sum_{j=1}^n X_i^j$ for $i\in C$.

%Let $C$ be a set of indices denoting indices of secret attributes.  The goal is to release (an approximate) of $F(X)=\sum_{k=1}^n \fixme{X_{jk}}$, the average value of a certain column $j \in [m]/C$ (for example, the average SAT score), while ensuring privacy against an adversary from inferring $g(X_i)=\sum_{k=1}^n \fixme{X_{ik}}$,
%the average value of any secret attribute $i\in C$ which could be, the average family income. 

To instantiate our framework, let $s_a^i$ denote the event that $g(X_i)=a$, which means the average value of column $X_i$ is $a$. If $g(X_i)$ has support $\mathcal{U}^i$, then the set of secrets is $S=\{s^i_a: a \in \mathcal{U}^i, i\in C\}$, and the set of secret pairs is $\pufQ=\{(s_a^i,s_b^i): a, b \in \mathcal{U}^i, a \neq b , i\in C\}$. Each $\theta \in \Theta$ is a distribution over $n$ i.i.d.~samples from an underlying multivariate Gaussian distribution with mean $(\mu_1,\ldots,\mu_m)^T$ and covariance matrix
$(V_{ij})$, $i, j \in [m]$,
%\wz{or if we have space}
%$$\begin{pmatrix} V_{11} & V_{12} & \ldots & V_{1m} \\ V_{21} & V_{22} &\ldots & V_{2m}\\ \vdots&\vdots &\vdots &\vdots\\V_{m1} &V_{m2}&\ldots &V_{mm} \end{pmatrix},$$
where $V_{ij}=V_{ji}$ is the covariance between $X_i$ and $X_j$ if $i\neq j$, and $V_{ii}$ is the variance of $X_i$. We note that the variable heights, weights and SAT score may not be Gaussian distributed in practice.
Hence, the choice of whether to use the Attribute-Private Gaussian Mechanism for Gaussian or non-Gaussian data
should be determined by the practitioner.

Suppose we want to guarantee $(\epsilon,\delta)$-dataset attribute privacy through the Attribute-Private Gaussian Mechanism. Then we need to first compute $\E{F(X)|s_a^i}$ and $\Var(F(X)|s_a^i)$ for each $i\in C$. Let~$j$ denote the index of the attribute averaged in $F(X)$. By the properties of a multivariate Gaussian distribution, the distribution of $F(X)$ conditional on $g(X_i)=a$ is Gaussian $\mathcal{N}(\bar{\mu}_a, \bar{V})$, where $\bar{\mu}_a=\mu_j+\frac{V_{ij}}{V_{ii}}(a-\mu_i)$ and $\bar{V}=\frac{1}{n}(V_{jj}-\frac{V_{ij}^2}{V_{ii}})$.
%$$\bar{\mu}_a=\mu_j+\frac{V_{ij}}{V_{ii}}(a-\mu_i),$$
%and
%$$\bar{V}=\frac{1}{n}(V_{jj}-\frac{V_{ij}^2}{V_{ii}}).$$
We define the diameter of $\mathcal{U}$ as $d(\mathcal{U})= \max_{a,b\in \mathcal{U}}\abs{a-b}$. The sensitivity is: $\Delta_i F=\max_{(s_a^i, s_b^i)\in \pufQ} \abs{\bar{\mu}_a-\bar{\mu}_b} =  \frac{V_{ij}}{V_{ii}}\max_{a,b \in \mathcal{U}}\abs{a-b}=\frac{V_{ij}}{V_{ii}}d(\mathcal{U})$.
%$$\Delta_i F=\max_{(s_a^i, s_b^i)\in Q} \abs{\bar{\mu}_a-\bar{\mu}_b} =  \frac{V_{ij}}{V_{ii}}\max_{a,b \in \mathcal{U}}\abs{a-b}=\frac{V_{ij}}{V_{ii}}d(\mathcal{U}).$$
To ensure $(\epsilon,\delta)$-dataset attribute privacy for protected attribute $X_i$, the variance of the Gaussian noise must be at least $(c\frac{V_{ij}d(\mathcal{U})}{V_{ii}\eps})^2-\frac{1}{n}(V_{jj}-\frac{V_{ij}^2}{V_{ii}})$ for $c=\sqrt{2\log (1.25/\delta)}$ as in Algorithm \ref{algo:gaussian}. Adding Gaussian noise with variance $\sigma^2=\max_{i\in C} \{ (c\frac{V_{ij}d(\mathcal{U})}{V_{ii}\eps})^2-\frac{1}{n}(V_{jj}-\frac{V_{ij}^2}{V_{ii}}) \}$
%\begin{equation*}
%\sigma^2=\max_{i\in C} \left\{ (c\frac{V_{ij}d(\mathcal{U})}{V_{ii}\eps})^2-\frac{1}{n}(V_{jj}-\frac{V_{ij}^2}{V_{ii}}) \right\},
%\end{equation*}
will provide $(\epsilon,\delta)$-dataset attribute privacy for all protected attributes.

We note that $\sigma^2$ is monotonically increasing with respect to $V_{ij}$.
That is, our Attribute-Private Gaussian Mechanism will add less noise to the output
if the query $F$ is about an attribute which has
a low correlation with the protected attributes.

So far we have discussed about the case when $\Theta$ only consists of one distribution, in order to show the impact of $V_{ij}$. For the general case, the sensitivity $\Delta_i F$ is $\max_{\theta\in\Theta} \frac{V_{ij}}{V_{ii}}d(\mathcal{U})$, and the noise is scaled with variance $\sigma^2=\max_{i\in C} \{ (c\max_{\theta\in\Theta}\frac{V_{ij}d(\mathcal{U})}{V_{ii}\eps})^2-\min_{\theta\in\Theta}\frac{1}{n}(V_{jj}-\frac{V_{ij}^2}{V_{ii}}) \}$.

\section{The Markov Quilt Mechanism for Distributional Attribute Privacy}\label{sec:distsecrets}

In this section we consider distributional attribute privacy, as introduced in Definition~\ref{def:distprivacy}, 
and develop a mechanism that satisfies this privacy definition.
Recall that in this setting, an analyst aims to release
$F(X)$ while protecting the realization of a random parameter $\phi_i$, which describes the conditional marginal distribution of the $i$th attribute, given the realization of all $\phi_k$ for $k\neq i$ for all other attributes.  This formalization implies that all (column) attribute vectors $X_i$ are mutually independent, conditional on the set of parameters $(\phi_1,\ldots,\phi_m)$. %In our setting, we assume $F$ is only computed 

\subsection{Attribute-Private Markov Quilt Mechanism}

We base our mechanism on the idea of a \emph{Markov Quilt}, which partitions a network of correlated random variables into those which are ``near'' ($X_N$) a particular variable $X_i$, and those which are ``remote'' ($X_R$).  Intuitively, we will use this to partition attributes into those which are highly correlated ($X_N$) with our sensitive attributes, and those which are only weakly correlated ($X_R$).

\begin{definition}[Markov Quilt]\label{def.markovquilt}
	A set of nodes $X_Q$ in a Bayesian network $G=(X,E)$ is a Markov Quilt for a node $X_i$ if deleting $X_Q$ partitions $G$ into parts $X_N$ and $X_R$ such that $X_i\in X_N$ and $X_R$ is independent of $X_i$ conditioned on $X_Q$.
\end{definition}

%\rc{left off editing here}

%Let $\phi_A \in \Phi_A$ to denote parameters of the underlying distribution for a set of attributes $X_A$. 
We quantify the effect that changing the distribution parameter $\phi_i$ of a sensitive attribute $X_i$ has on a set of distribution parameters $\phi_A$ (corresponding to a set of attributes $X_A$) using the \emph{max-influence}. Since attributes are mutually independent conditioned on the vector $(\phi_1,\ldots,\phi_m)$, the max-influence is sufficient to quantify how much a change of all values in attribute $X_i$ will affect the values of $X_A$.
If $\phi_i$ and $\phi_A$ are independent, then $X_A$ and $X_i$ are also independent, and the max-influence is 0.

%We quantify how much changing the distribution parameter of a sensitive attribute $X_i$, $\phi_i$, will affect~$\phi_A$ by the max-influence defined as follows.
%Recall that $\Theta$ is a set of possible distributions $\theta$ over the database $X$
%where $\theta$ describes the conditional marginal distributions for all attributes.
%%The dependence is described by the distribution class $\Theta$.
%Since attributes are mutually independent conditioned on the vector
%of $(\phi_1,\ldots,\phi_m)$, the max-influence is sufficient to quantify how much a change of all values in $X_i$ will affect values of other attributes $X_A$.

\begin{definition}\label{def:maxInfD} The \emph{max-influence} of an attribute $X_i$ on a set of attributes $X_A$ under $\Theta$ is:
\[ e_\Theta(X_A|X_i)=\sup_{\theta\in\Theta}\max_{\phi_i^a,\phi_i^b\in\Phi_i}\max_{\phi_A\in\Phi_A} \log \frac{P(\phi_A|\phi_i^a,\theta)}{P(\phi_A|\phi_i^b,\theta)}. \]
%\begin{equation*}
%\[ e_\Theta(X_A|X_i)=\sup_{\theta\in\Theta}\max_{\phi_i^a,\phi_i^b\in\Phi_i}\max_{\phi_A\in\Phi_A} \log \frac{P(\phi_A|\phi_i^a,\theta)}{P(\phi_A|\phi_i^b,\theta)}. \]
%\end{equation*}
\end{definition}

%Here $\Phi_i$ is the domain of $\phi_i$, and $\Phi_A$ is the domain of $\phi_A$. 
%The max-influence is the maximum max-divergence between the distributions $(\phi_A|\phi_i^a,\theta)$ and $(\phi_A|\phi_i^b,\theta)$ where the maximum is taken over any pair $(\phi_i^a, \phi_i^b)\in \Phi_i \times \Phi_i$ and any $\theta\in \Theta$. 

 %, since the distribution of $X_A$ conditioned on $\phi_i$ is its own distribution $P(X_A)$.
%\olya{Should it be $\phi_i$ and $\phi_A$ that are independent? and as a result $X_A$ and $X_i$ are}

%\textcolor{blue}{Let $\mu_a=P(\phi_A|\phi_i^a,\theta)$ and $\mu_b=P(\phi_A|\phi_i^b,\theta)$, the max-influence is the maximum max-divergence between $\mu_a$ and $\mu_b$.}

%\todo{define $\Delta_i F$}

The sensitivity of $F$ with respect to a set of attributes $A \subseteq [m]$, denoted $\Delta_A F$, is defined as the maximum change that the value of $F(X)$ caused by changing all columns $X_A$. Formally, we say that two datasets $X, X'$ are \emph{$A$-column-neighbors} if they are identical except for the columns corresponding to attributes in $A$, which may be arbitrarily different.  Then $\Delta_A F = \max_{X,X' \; A\text{-column-neighbors}} \abs{F(X)-F(X')}$. 
Although changing $X_A$ may lead to changes in other columns,
these changes are governed by the max influence, and will not affect attributes that are nearly independent of $X_i$.

Observe that the event that $X_R$ and $X_i$ are independent conditional on $X_N$
is equivalent to the event when $\phi_R$ and $\phi_i$ are independent conditional on $\phi_N$, which is why we can define the Markov Quilt based on $X_i$. However,
since the distribution of $X_i$s are governed by $\phi_i$s, the max-influence score must be computed
using $\phi_i$s rather than $X_i$s.

%\begin{definition}[sensitivity] We define the sensitivity of the query 
%\begin{equation*}
%\Delta_i F = \max_{(X_i,X_i')} \abs{F(X_i)-F(X_i')}
%\end{equation*}
%\end{definition}

%We assume that we want to output a query value on the database $F(X)$.
%The sensitivity of $F$ with respect to $X_i$, denoted by $\Delta_i F$, is defined as the maximum amount of the change of $F(X)$ if the entire column $X_i$ changes.
%Though, the change in the $X_i$ column may lead to changes in other columns,
%these changes are governed by the structure of the Bayesian network and will not affect attributes that are almost independent of $X_i$.

%\fixme{Therefore, same as~\cite{song2017pufferfish}, we define Markov Quilt based on $X_i$. However,
%since the distribution of $X_i$'s are governed by $\phi_i$'s, max-influence score (Definition~\ref{def:maxInfD}) is different as it has to be computed
%using $\phi_i$'s and not $X_i$'s as is the case in~\cite{song2017pufferfish}.}
%Therefore, the Markov Quilt for $X_i$ in our definition is the same set as in \cite{song2017pufferfish}.
%We note that although the Markov Quilt set is the same, our max-influence score is different. 

\paragraph{The mechanism.}
We extend the idea of the Markov Quilt Mechanism in \cite{song2017pufferfish} to the attribute privacy setting as follows. Let $A \subseteq [m]$ be a set of attributes over which $F$ is computed. For example, $F$ may compute the average of a particular attribute or a regression on several attributes.
At a high level, we add noise to the output of $F$ scaled based on the sensitivity of $F$ with respect to $X_Ns$.
However, when computing the sensitivity of $F$ we only need to consider sensitivity of $F$ with respect to ${A\cap N}$,
i.e.,  the queried set of attributes $A$ that are in the ``nearby'' set of the protected attribute.
If the query $F$ is about attributes that are all in the ``remote'' set $X_{R_i}$ and the max-influence on the corresponding Markov quilt is less than the privacy parameter $\eps$, then
$\Delta_{A\cap N} F$ is simply 0 and the mechanism will not add noise to the query answer.

	\begin{algorithm}[tbh]
	\centering
		\caption{Attribute-Private Markov Quilt Mechanism, APMQM($X, F, A, C,\{S, \pufQ, \Theta \}, \epsilon$) for distributional attribute privacy.}
		\begin{algorithmic}
			\State \textbf{Input:} dataset $X$, query $F$, index set of queried attributes $A$, index set of sensitive attributes $C$, framework $\{S, \pufQ, \Theta\}$, privacy parameter $\epsilon$.
			\State \text{for} each $i\in C$ \textbf{do}
			\State \indent Set $b_i=\Delta_{A} F/\eps$.
			\State \indent Set $G_i:=\{(X_{Q},X_{N},X_{R}): e_\Theta(X_{Q}|X_i)\le \eps\}$ to be all possible Markov quilts of $X_i$ with max-influence less than $\eps$.
			\State \indent \textbf{if} $G_i \neq \emptyset$ \textbf{do}
			\State \indent \indent \textbf{for} each $(X_{Q},X_{N},X_{R}) \in G_i$ \textbf{do}
			\State \indent \indent \indent \textbf{if} $ \Delta_{A\cap N} F/(\eps-e_\Theta(X_{Q}|X_i)) \le b_i$ \textbf{then} %\wz{do we need subscript for N?}
			\State \indent \indent \indent \indent Set $b_i = \Delta_{A\cap N} F/(\eps-e_\Theta(X_{Q}|X_i))$.
			\State Sample $Z\sim\Lap(\max_{i\in C}b_i)$.
			\State Return $F(X)+Z$			
		\end{algorithmic}
		\label{alg:quilt}
	\end{algorithm}
	%	\end{minipage}
%}

%
%
%{\centering
%	%	\begin{minipage}{\linewidth}
%	\begin{algorithm}[H]
%		\caption{Attribute-Private Markov Quilt Mechanism $\mathcal{M}(X)$ with parameters ($F, A, C,\{S, Q, \Theta \}, \epsilon$) for distributional attribute privacy.}
%		\begin{algorithmic}
%			\State \textbf{Input:} database $X$, query $F$, index set of queried attributes $A$, index set of sensitive attributes $C$, framework $\{S, Q, \Theta\}$, privacy parameter $\epsilon$.
%			\For {each $i\in C$}
%			\State Set $\Delta_i=\Delta_{A} F$.
%			\State Set $G:=\{(X_Q^i,X_N^i,X_R^i): e_\Theta(X_Q^i|X_i)\le \eps\}$, which are all possible Markov quilts partition with max-influence less than $\eps$.		
%			\If {$G \neq \emptyset$}
%			\For {each $(X_Q^i,X_N^i,X_R^i) \in G$}
%			\If {$\Delta_i\ge \Delta_{A\cap N} F$}\wz{do we need subscript for N?}
%			\State Set $\Delta_i=\Delta_{A\cap N} F$.
%			\EndIf
%			\EndFor
%			\EndIf
%			\EndFor
%			\State Sample $Z\sim\Lap(\max_{i\in C}\Delta_i/(\eps-\Theta(X_Q|X_i)))$.
%			\State Return $F(X)+Z$			
%		\end{algorithmic}
%		\label{alg:quilt}
%	\end{algorithm}
%	%	\end{minipage}
%}

\begin{restatable}{theorem}{markov}\label{thm.markov}
%\begin{theorem}\label{thm.markov}
	The Attribute-Private Markov Quilt Mechanism APMQM($X,F, A, C,\{S, \pufQ, \Theta \}, \epsilon$) in Algorithm~\ref{alg:quilt} is $(\eps,0)$-distributional attribute private.
%\end{theorem}
\end{restatable}

\begin{proof}
	First, consider the case when $G=\emptyset$, which means there are no Markov quilt partitions such that the max-influence score is less than $\eps$.
	In this case, for a fixed secret attribute $X_i$, the mechanism will simply add Laplace noise scaled with $\Delta_A F$. The privacy follows from Laplace Mechanism in differential privacy. 
	
	%When $G\neq \emptyset$, \fixme{it suffices to consider the ratio of probability distribution functions, as we can integrate over any subset $T$. }\rc{can we rephrase to make more clear? we are considering a single element $w$ rather than a set of outcomes $T$}
	Let us consider the case when $G\neq \emptyset$. Below we bound the probability distribution for a single outcome $w$.
	For the set of outcomes $T$, the proof can be extended to bound the integral of the probability distribution over the set $T$.
	We fix any pair of secrets $(s_a^i,s_b^i)\in \pufQ$ for a fixed secret attribute $X_i$ under  any $\theta \in \Theta$.
	Here $s_a^i$ denotes the event that $\phi_i=a$, and we consider the more general case later. Let $Z$ be the Laplace noise as generated in Algorithm~\ref{alg:quilt}. If there exists a Markov quilt for $X_i$, $i\in C$, with max-influence score less than $\eps$, then for any attribute $X_j \in A$:
	% that the query $F$ is about 
	\begin{align}
	&\max_{a,b}\frac{P(\mathcal{M}(X)=w|\phi_i=a,\theta)}{P(\mathcal{M}(X)=w|\phi_i=b,\theta)} \notag\\
	&=\max_{a,b}\frac{P(F(X)+Z=w|\phi_i=a,\theta)}{P(F(X)+Z=w|\phi_i=b,\theta)} \notag\\
	=&\max_{a,b}\max_{{R\cup Q}}\frac{P(F(X)+Z=w|\phi_i=a,\phi_{R\cup Q}=\bar{v},\theta)}{P(F(X)+Z=w|\phi_i=b,\phi_{R\cup Q}=\bar{v},\theta)}\frac{P(\phi_{R\cup Q}=\bar{v}|\phi_i=a,\theta)}{P(\phi_{R\cup Q}=\bar{v}|\phi_i=b,\theta)} \notag\\
	=&\max_{a,b}\max_{{R\cup Q}}\frac{P(F(X)+Z=w|X_{R\cup Q}=x_{R\cup Q}, \phi_i=a,\theta)} {P(F(X)+Z=w|X_{R\cup Q}=x_{R\cup Q},\phi_i=b,\theta)}  \cdot  \frac{P(X_{R\cup Q}=x_{R\cup Q}|\phi_{R\cup Q}=\bar{v},\theta)}{P(X_{R\cup Q}=x_{R\cup Q}|\phi_{R\cup Q}=\bar{v},\theta)} \notag\\&%\frac{P(X_N=(X^j_N)|\phi_i=a,\theta)}{P(X_N=(X^j_N)|\phi_i=b,\theta)}
	 \qquad \qquad \cdot \frac{P(\phi_{R\cup Q}=\bar{v}|\phi_i=a,\theta)}{P(\phi_{R\cup Q}=\bar{v}|\phi_i=b,\theta)},\label{eq.mq1}
	\end{align}
	where the final equality Equation \eqref{eq.mq1} follows from the independence between $X_{R\cup Q}$ and $\phi_i$ given $\phi_{R\cup Q}=\bar{v}$, and $x_{R\cup Q}$ denotes a realization of the $X_{R\cup Q}$ columns. For a fixed $X_{R\cup Q}$, $F(X)$ can vary by at most $\Delta_{A\cap N} F$, and therefore, the first ratio is bounded by $\exp(\eps-\exp(e_\Theta(X_Q|X_i)))$. The second ratio in Equation \eqref{eq.mq1} is 1, and the third ratio in Equation \eqref{eq.mq1} is bounded by $\exp(e_\Theta(X_Q|X_i))$. Then Equation \eqref{eq.mq1} is bounded above by $\exp(\eps-\exp(e_\Theta(X_Q|X_i))) \exp(e_\Theta(X_Q|X_i)) = \exp(\eps)$. For the general case when $s^i_a := \mathbbm{1}[\phi_i \in \Phi^i_a] : \Phi^i_a \subset \Phi^i$, similarly, for any $\Phi^i_a$ and $\Phi^i_b$, we have, 
	\begin{align}
	&\frac{P(\mathcal{M}(X)=w|\phi_i \in \Phi^i_a,\theta)}{P(\mathcal{M}(X)=w|\phi_i  \in \Phi^i_b,\theta)} \notag\\
	\le&\max_{{R\cup Q}}\frac{P(F(X)+Z=w|X_{R\cup Q}=x_{R\cup Q}, \phi_i \in \Phi^i_a,\theta)} {P(F(X)+Z=w|X_{R\cup Q}=x_{R\cup Q},\phi_i \in \Phi^i_b,\theta)} \frac{P(\phi_{R\cup Q}=\bar{v}|\phi_i \in \Phi^i_a,\theta)}{P(\phi_{R\cup Q}=\bar{v}|\phi_i\in \Phi^i_b,\theta)}\\
	\le&\max_{{R\cup Q}} \max_{a\in \Phi^i_a, b\in \Phi^i_b}\frac{P(F(X)+Z=w|X_{R\cup Q}=x_{R\cup Q}, \phi_i=a,\theta)} {P(F(X)+Z=w|X_{R\cup Q}=x_{R\cup Q},\phi_i=b,\theta)} \max_{a\in \Phi^i_a, b\in \Phi^i_b} \frac{P(\phi_{R\cup Q}=\bar{v}|\phi_i=a,\theta)}{P(\phi_{R\cup Q}=\bar{v}|\phi_i=b,\theta)}. \label{eq.mq2}
	\end{align}
	The first ratio in Equation \eqref{eq.mq2} is bounded by $\exp(\eps-\exp(e_\Theta(X_Q|X_i)))$ and the second ratio is bounded by $\exp(e_\Theta(X_Q|X_i))$, so Equation \eqref{eq.mq2} is bounded above by $\exp(\eps)$, and the theorem follows.

	%\olya{certain is $n$ realizations? or just realizations in $X$}\wz{what's the difference? I mean fix all entries in the columns}
	%\olya{is $v$ a vector, may be use $\bar{v}$?}
	%\olya{I still do not follow the notation with $^j$}
\end{proof}

%\fixme{The attribute-private Markov Quilt Mechanism does not compose in general, since different Markov quilts could be used depending on which attributes appear in the queries.
%As a result, the correlation between attributes in different Markov quilts may violate their privacy guarantees when combined.}
%However, if all queries use the same Markov quilt for each secret attribute, then privacy does compose. For example, if all queries are about the same subset of attributes. We note the Markov Quilt Mechanism in \cite{song2017pufferfish} has a similar property: its composition is based on using the same active Markov quilt.

%\begin{example}
%\if 0
%	\begin{figure}[h]
%		\centering
%		\includegraphics[width=3in]{example.png}
%		\caption{Bayesian Network of five attributes.}
%	\end{figure}\label{fig:example}
%\fi
 \begin{wrapfigure}{r}{0.35\textwidth}
  \centering
    \includegraphics[width=0.34\textwidth]{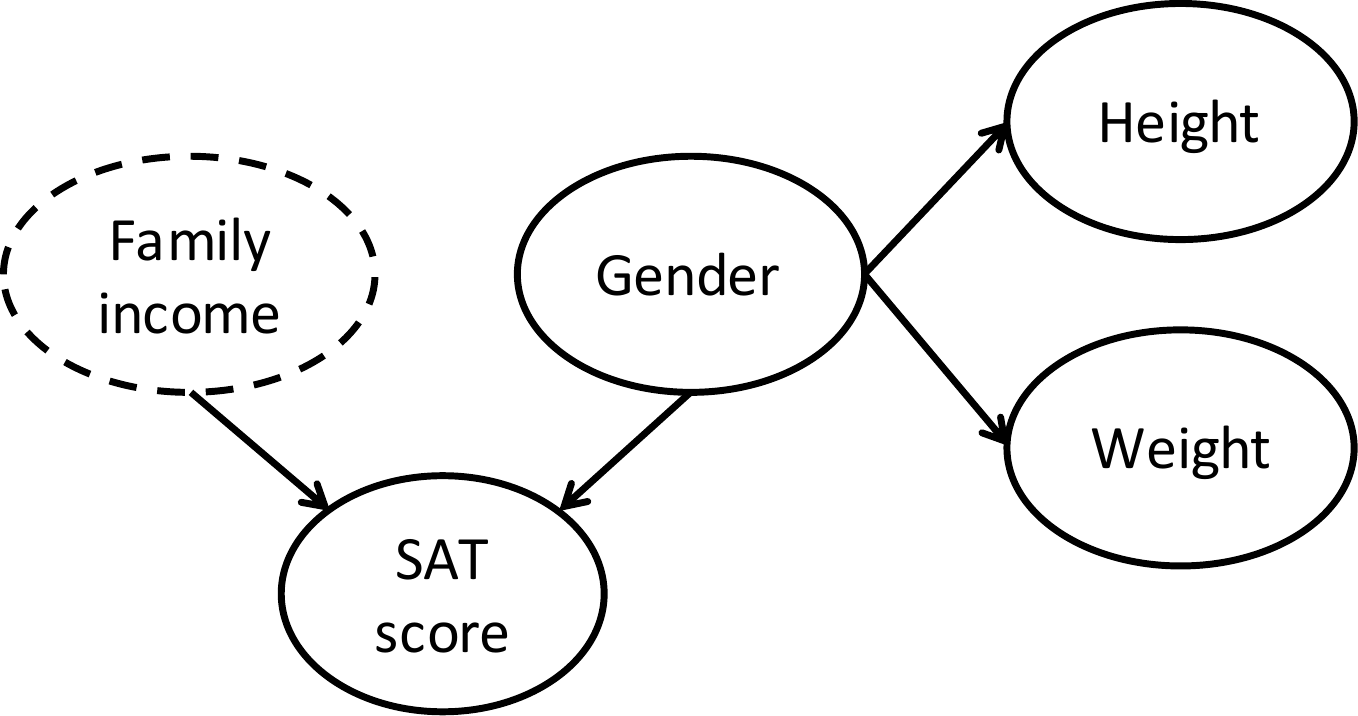}
\caption{\small Bayesian Network of five attributes where income is a sensitive attribute.}
\label{fig:example}
\end{wrapfigure}

\begin{example} Consider a dataset that consists of students' SAT scores $X_\sat$, heights $X_\hgt$, weights $X_\wgt$, gender $X_\gen$, and their family income $X_\inc$, where these variables form a Bayesian network as in Figure \ref{fig:example}.
The school wishes to release the number of students that are taller than $5'6''$,
while protecting the distribution of family income of their students with privacy parameter~$\eps$.
In this case, $C=\{\inc\}$, $A=\{\hgt\}$ and $F(X)=\sum_{j=1}^n \mathbbm{1}[X_\hgt^j>5'6'']$.
Consider a Markov quilt for $X_\inc$: ${Q}=\{\gen\}, N=\{\inc,\sat\}, R=\{\hgt,\wgt\}$.
 %Suppose that under $\Theta$, the max-influence score between $X_i$ and $X_Q$ is less than $\eps$. 
 Then $A\cap N=\emptyset$, so we can safely release $F(X)=\sum_{i=1}^n\mathbbm{1}[X_\hgt^j>5'6'']$ without additional noise.

Next consider the case when the school wishes to release the number of students that are taller than $5'6''$ and have SAT score $>1300$.
Then, $F(X)=\sum_{j=1}^n\mathbbm{1}[(X_\hgt^j>5'6'') \wedge (X_\sat^j>1300)]$ and $A=\{\hgt,\sat\}$.
%Consider another Markov quilt of $X_\inc$ where $X_Q = \{X_\sat\}$ and its max-influence score is greater than $\eps$.
In this case we can still use the same Markov quilt as before, but now $A\cap N = \{\sat\}$. 
%The sensitivity $\Delta_{A\cap N} F=\Delta_{\{\sat\}} F$ and the worst case is when students' SAT scores are all below 1300 or all above 1300.
The mechanism will add Laplace noise scaled with $\Delta_{\{\sat\}} F/(\eps-e_\Theta(X_\gen|X_\inc))$.
%\olya{I still do not understand why we need to mention another Markov quilt here and we did not need to mention it in the previous paragraph.}
\end{example}

%\olya{Wanrong is this even close to truth?}\wz{great!}
It is instructive to contrast the above mechanism to the Markov Quilt Mechanism of \cite{song2017pufferfish}, presented fully in Appendix \ref{app.prelim}. The most important difference is that the mechanism in~\cite{song2017pufferfish} was not designed to guarantee
attribute privacy. It provides
privacy of the values $X^j_i$ but does not protect the distribution
from which $X^j_i$ is generated. This difference in high-level goals leads to three key technical differences. Firstly, the definition of max-influence in~\cite{song2017pufferfish}
measures influence of a \emph{variable value on values of other variables}.
This is insufficient when one wants to protect distributional information,
as $X_i$ may take a range of values while still following a particular
distribution (e.g., hiding the gender of an individual
in a dataset
vs.~hiding the proportion of females to males in this dataset.)
Secondly, while it is natural to consider $L$-Lipschitz
functions to bound sensitivity when one value changes (as is done in \cite{song2017pufferfish}), this
is not applicable to settings where the distribution of data changes, since this may change all values in a column.
As a result, we do not restrict $F$ in this way.
Finally, the mechanisms themselves are different as~\cite{song2017pufferfish} consider answering query $F$ over all attributes
of an individual. As a result, they need to consider sensitivity of a function to all the ``nearby'' attributes.
In contrast, we only consider sensitivity of those ``nearby'' attributes that happen to be in the query (i.e., those in $A$).
%\olya{but how would we compute $F$...}

%!TEX root = attributePrivacy.tex

\section{The Wasserstein Mechanism for General Attribute Privacy}\label{s.wass}
%\wz{maybe this should go to related work}

%\rc{move some text from prelims about wasserstein mech here?}

The Wasserstein Mechanism \cite{song2017pufferfish} (Algorithm \ref{alg.wass}) is a general mechanism for satisfying Pufferfish privacy; Algorithm \ref{alg.wass} is $(\eps,0)$-Pufferfish private for any instantiation of the Pufferfish framework \cite{song2017pufferfish}.
It defines sensitivity of a function $F$ as the maximum Wasserstein distance $W_\infty$ between the distributions of $F(X)$ under two different realizations of secrets $s_i$ and $s_j$ for $(s_i,s_j) \in \pufQ$, and then outputs $F(X)$ plus Laplace noise that scales with this sensitivity. 
The distance metric $W_\infty$ denotes the $\infty$-Wasserstein distance between two probability distributions, formally defined below.

%Although this mechanism works in general for any instantiation of the Pufferfish framework, computing Wasserstein distance for all pairs of secrets is computationally expensive, and will typically not be feasible in practice.  

%\cite{song2017pufferfish} proposed the Wasserstein Mechanism and showed that it satisfies $(\eps,0)$-Pufferfish private.
%The mechanism is given in Algorithm~\ref{alg.wass} where

\begin{definition}[$\infty$-Wasserstein distance, $W_\infty$]\label{def.winfty} Let $\mu, \nu$ be two probability distributions on $\mathbb{R}$,\footnote{In general, Wasserstein distance can be defined on any metric space. We will use it only over the real numbers with the Euclidean metric.}
	%\rc{is this $\mathbb{R}$ or an arbitrary set $\mathcal{R}$?}\olya{I thought it was support set of $\mu$ and $\nu$. It is $\mathbb{R}$ in Kamalika's paper and I cannot find a good reference otherwise. Let's just say "be two probability distributions on a set $\mathcal{R}$"} \rc{I guess $|x-y|$ has to be well defined. My bet is that it can be defined over any metric space, but we will use it only over $\mathbb{R}$, which allows our distance metric to be absolute value.  We could say $\mathbb{R}$ here in our definition (which would make things easier to read) and have a footnote saying it can be defined more generally over any metric space?}
	 and let $\Gamma(\mu,\nu)$ be the set of all joint distributions with marginals $\mu$ and $\nu$. The $\infty$-Wasserstein distance between $\mu$ and $\nu$ is defined as :
\begin{equation}
W_\infty(\mu,\nu)=\inf_{\gamma\in \Gamma(\mu,\nu)}\max_{(x,y)\in support(\gamma)}\abs{x-y}.
\end{equation}
\end{definition}

The $\infty$-Wasserstein distance is closely related to optimal transportation. Each $\gamma \in \Gamma(\mu,\nu)$ can be viewed as a way to shift probability mass between $\mu$ and $\nu$, and the cost is $\max_{(x,y)\in support(\gamma)}\abs{x-y}$. For discrete distributions, the $\infty$-Wasserstein distance is the minimum of the maximum distance that any probability mass moves to transform $\mu$ to $\nu$. 

%\rc{add definition. include general statement and then (either in def or in text after) give simpler version for discrete dist since we use that in our examples.}

%{\centering
%	\begin{minipage}{\linewidth}
		\begin{algorithm}[H]
		\centering
			\caption{Wasserstein Mechanism ($X, F, \{S, \pufQ, \Theta \}, \epsilon$) \cite{song2017pufferfish}}
			\begin{algorithmic}
				\State \textbf{Input:} dataset $X$, query $F$, Pufferfish framework $\{S, \pufQ, \Theta \}$, privacy parameter $\epsilon$
				\For {all $(s_i,s_j)\in \pufQ$ and all $\theta\in \Theta$ such that $P(s_i|\theta)\neq 0$, and $P(s_j|\theta)\neq 0$}
				\State Set $\mu_{i,\theta}=P(F(X)|s_i,\theta)$, $\mu_{j,\theta}=P(F(X)|s_j,\theta)$.
				\State Calculate $W_\infty(\mu_{i,\theta},\mu_{j,\theta}).$
				\EndFor
				\State Set $W=\sup_{(s_i,s_j)\in \pufQ, \theta\in \Theta}W_\infty(\mu_{i,\theta},\mu_{j,\theta}).$
				\State Sample $Z\sim\Lap(W/\epsilon)$.
				\State Return $F(X)+Z$
			\end{algorithmic}
			\label{alg.wass}
		\end{algorithm}
%	\end{minipage}
%}

Since our framework is an instantiation of Pufferfish privacy, the Wasserstein Mechanism provides a general way to protect either dataset attribute privacy or distributional attribute privacy, when instantiated with the appropriate Pufferfish framework $(S,\pufQ,\Theta)$. This is stated formally in Theorem \ref{thm.wass} and illustrated in Examples \ref{ex.wassdata} and \ref{ex.wassdist} below.

%We summarize these results below.

\begin{theorem}\label{thm.wass}
	The Wasserstein Mechanism ($X, F, \{S, \pufQ, \Theta \}, \epsilon$) in Algorithm~\ref{alg.wass} is $(\eps,0)$-dataset attribute private and $(\eps,0)$-distributional attribute private.
\end{theorem}

%\begin{theorem}
%	The Wasserstein Mechanism ($X, F, \{S, \pufQ, \Theta \}, \epsilon$) in Algorithm~\ref{alg.wass} is $(\eps,0)$-distributional attribute private.
%%\end{theorem}
%\end{theorem}

Despite the general purpose nature of the Wasserstein Mechanism for achieving attribute privacy, it is known that computing Wasserstein distance is computationally expensive \cite{atasu2019linear}. Instantiating Algorithm \ref{alg.wass} to satisfy attribute privacy may require computing Wasserstein distance for exponentially many pairs of secrets, one for each subset of values of $g_i(X_i)$ or $\phi_i$. This motivates our study of the Attribute-Private Gaussian Mechanism (Algorithm~\ref{algo:gaussian}) and the Attribute-Private Markov Quilt Mechanism (Algorithm~\ref{alg:quilt}), which are both computationally efficient for practical use.

In some special cases, the Attribute-Private Wasserstein Mechanism may be a feasible option. For example, when computing $\infty$-Wasserstein distance between a pair of 1-dimensional distributions as in Definition \ref{def.winfty}, then $W_{\infty}$ can be computed efficiently \cite{LLL18}. While there exist efficient approximations to $\infty$-Wasserstein distance, any approximation used in the Attribute-Private Wasserstein Mechanism must always be an overestimate of $W_{\infty}$  to ensure that sufficient noise is added to guarantee privacy. Efficiently computable upper bounds on $W_{\infty}$ exist under certain technical conditions on the distributions~\cite{GG20}. In both the case of 1-dimensional distributions and approximations of $\infty$-Wasserstein distance, if the instantiation of Pufferfish privacy requires computing $W_{\infty}$ over a polynomial number of distributions (i.e., polynomially many pairs of secrets and polynomially many possible distributions $\theta$), then the Attribute-Private Wasserstein Mechanism (or an approximate version of the mechanism) can be run efficiently.
%\olya{Rachel, I like this paragraph. I was only unsure what "both of these cases" mean. Probably both of definitions? May be good to clarify as i though case 1. 1 dimension and case 2 approximation"}\rc{changed it. Is this more clear now?}

\begin{example}[Wasserstein Mechanism for Dataset Attribute Privacy]\label{ex.wassdata}
	Consider two binary attributes $X_1$ and $X_2$, where $X_1$ is the non-sensitive attribute and $X_2$ is the sensitive attribute. Suppose the dataset contains data from four people, and let the underlying distribution and dependence between $X_1$ and $X_2$ is characterized by the following probability distributions:
	\begin{equation}\label{eq.depend}
	P(X_1=1|X_2=1)=p_1  \text{ and } P(X_1=1|X_2=0)=p_2,
	\end{equation}
%	\begin{align}\label{eq.depend}
%		P(X_1=1|X_2=1)&=p_1  \\
%		P(X_1=1|X_2=0)&=p_2.
%	\end{align}
	Suppose $0.4\le p_1\le 0.6$ and $0.4 \le p_2 \le 0.6$. 
	The analyst wishes to release the summation of $X_1$, $F(X)=\sum_{j=1}^4 X_1^j$, while protecting the summation of $X_2$, so $g(X_2)=\sum_{j=1}^4 X_2^j$. To instantiate our framework, let $s_a^2$ denote the event that $g(X_2)=a$. The support of $g(X_2)$ is $\mathcal{U}=\{0,1,2,3,4\}$. Then the set of secrets is $S=\{s_a^2: a\in \mathcal{U}\}$, and the set of secret pairs is $\pufQ=\{(s_a^2,s_b^2): a, b \in \mathcal{U}, a\neq b\}$. Each $\theta\in \Theta$ is a certain pair of $p_1$ and $p_2$ such that $0.4\le p_1\le 0.6$ and $0.4 \le p_2 \le 0.6$.

Consider the pair of conditional probabilities $\mu_{a,\theta}=P(F(X)=\cdot|s_a^2,\theta)$ and $\mu_{b,\theta}=P(F(X)=\cdot|s_b^2,\theta)$. The worst case Wasserstein distribution between the pair of conditional probability distributions is reached when $p_1=0.4$, $p_2=0.6$, and $a=0$, $b=4$ when the two conditional probabilities differ the most. We list the conditional probabilities for this case in Table \ref{table}.
	\begin{table}[h]
		\centering
		\begin{tabular}{|c|c|c|c|c|c|}
			\hline
			$j$&0&1&2&3&4\\
			\hline
			$P(F(X)=j|g(X_2)=0)$&0.0256&0.1536&0.3456&0.3456&0.1296\\
			\hline
			$P(F(X)=j|g(X_2)=4)$&0.1296&0.3456&0.3456&0.1536&0.0256\\
			\hline
		\end{tabular}
		\caption{Conditional probability distributions under two extreme secrets for dataset attribute privacy}\label{table}
	\end{table}
	
Here, the Wasserstein distance $W_\infty(P(F(X)|g(X_2=0)), P(F(X)|g(X_2=4))= 1$, since the optimal transportation is moving the mass from $1$ to $2$ and $4$ to $3$, and the Wasserstein mechanism will add $\Lap(1/\eps)$ noise to $F(X)$. 

%\rc{clarify text} \wz{leave for Olya} \rc{this text seems more clear now that we have better definitions etc. Leaving this comment in case Olya wants to add anything}
The mechanism with group differential privacy would add $\Lap(4/\eps)$, which gives worse utility. We note that the noise we add depends largely on the underlying distribution class $\Theta$. For example, when $\Theta=\{0.3\le p_1, p_2\le 0.7\}$, the worst case Wasserstein distance is 2 and the mechanism will add $\Lap(2/\eps)$ noise to $F(X)$. When $\Theta=\{0\le p_1, p_2\le 1\}$, the worst case Wasserstein distance is 4, and the Wassertein mechanism will add the same amount of noise as group differential privacy.
	\end{example}
	
	\begin{example}[Wasserstein Mechanism for Distributional Attribute Privacy]\label{ex.wassdist}
	%\rc{while this presentation of starting with $p_i$ and then computing $\phi_i$ is technically correct, it's not necessarily the most clear presentation.  Can we start with presenting $\phi_i$ and then tie it back to $p_i$ as needed?}\wz{I think it's more intuitive to start from $p_i$, since $\phi_i$ is conditional marginal distribution. One cannot define $\phi_i$ without defining the dependence between the two variables.} \rc{okay}
	
	Consider the same setting as Example \ref{ex.wassdata}: a dataset of four people with two binary attributes $X_1$ and $X_2$, where $X_1$ is non-sensitive and $X_2$ is sensitive.  Let the underlying distribution and dependence between realized attributes $X_1$ and $X_2$ still be governed by \eqref{eq.depend}, and for simplicity fix $p_1=0.4$ and $p_2=0.6$. In the setting of distributional attribute privacy, we are interested in the conditional marginal distribution parameters of $X_i$ given the \emph{parameter} for $X_j$, rather than the realization of $X_j$. We denote the Bernoulli distribution parameter for $X_1$ and $X_2$ as $\phi_1$ and $\phi_2$, respectively. According to ~\eqref{eq.depend}, we have $\phi_1=0.4\phi_2+0.6(1-\phi_2)=0.6-0.2\phi_2$.
	
	The analyst wishes to release the summation of $X_1$, $F(X)=\sum_{j=1}^4 X_1^j$, while protecting the distribution parameter $\phi_2$ for $X_2$. To instantiate our framework, we let $s_a^2$ denote the event that $\phi_2=a$, and we suppose the support of $\phi_2$ is $\Phi^2=[0.2,0.8]$. The set of secrets is $S=\{s_a^2: a\in \Phi^2\}$, and the set of secret pairs is $\pufQ=\{(s_a^2,s_b^2): a, b \in \Phi^2, a\neq b\}$.  Each $\theta\in\Theta$ is a certain pair of $\phi_1$ and $\phi_2$ such that $\Phi^2=[0.2,0.8]$ and $\phi_1=0.6-0.2\phi_2$. 
	
	In this case, the support for $\phi_1$ is $[0.44,0.56]$. Consider the pair of conditional probabilities $\mu_{a,\theta}=P(F(X)=\cdot|s_a^2,\theta)$ and $\mu_{b,\theta}=P(F(X)=\cdot|s_b^2,\theta)$. The worst case Wasserstein distribution between the pair of conditional probability distributions is reached when $a=0.8$ and $b=0.2$ when the two conditional probabilities differ the most. We list the conditional probabilities for this case in Table \ref{table2}.
	
	\begin{table}[h]
		\centering
		\begin{tabular}{|c|c|c|c|c|c|}
			\hline
			$j$&0&1&2&3&4\\
			\hline
			$P(F(X)=j|\phi_2=0.8)$&0.0983&0.3091&0.3643&0.1908&0.0375\\
			\hline
			$P(F(X)=j|\phi_2=0.2)$&0.0375&0.1908&0.3643&0.3091&0.0983\\
			\hline
		\end{tabular}
		\caption{Conditional probability distributions under two extreme secrets for distributional attribute privacy}\label{table2}
	\end{table}
	
	The Wasserstein distance $W_\infty(P(F(X)|\phi_2=0.8), P(F(X)|\phi_2=0.2)= 1$, since the optimal transportation is moving the mass from $1$ to $2$ and $4$ to $3$, and the Wasserstein mechanism will add $\Lap(1/\eps)$ noise to $F(X)$. 
\end{example}

\bibliography{private,olyarefs}
\bibliographystyle{alpha}

\newpage

\appendix

%!TEX root = attributePrivacy.tex

\section{Additional Preliminaries}\label{app.prelim}

In this appendix we review the Markov Quilt Mechanism of \cite{song2017pufferfish} for satisfying Pufferfish privacy.  This algorithm assumes that the entries in the input database $Y$ form a Bayesian Network (Definition \ref{def.bayesnet}).  These entires could either be: (1) the multiple attributes of a single record when the database contained only one record, or (2) the attribute values across multiple records for a single-fixed attribute when the database contained multiple attributes.  Hence, the original Markov Quilt Mechanism could not accommodate correlations across multiple attributes in multiple records, as we study in this work.

%\begin{definition}[Bayesian Networks]
%	A Bayesian network is described by a set of variables $Y=\{Y_1,\ldots, Y_n\}$ and a directed acyclic graph $G=(Y,E)$ whose vertices are variables in $Y$. The probabilistic dependence on  $Y$ included by the network can be written as:
%	\[ \Pr(Y_1,\ldots,Y_n)=\Pi_{i=1}^n \Pr(Y_i|\parent(Y_i)).\] 
%\end{definition}

The following definition measures influence of a \emph{variable value on values of other variables}.  One can compare this to Definition~\ref{def:maxInfD} used in the Attribute-Private Markov Quilt Mechanism, which instead measures max-influence of a parameter of the probability distribution of a variable on the distribution of parameters of other variables, as is needed in the attribute privacy setting.

\begin{definition}[Variable-Max-Influence~\cite{song2017pufferfish}]
	The maximum influence of a variable $Y_i$ on a set of variables $Y_A$ under $\Theta$ is:
	$$e^v_\Theta(Y_A|Y_i)=\sup_{\theta\in\Theta}\max_{a,b\in\mathcal{Y}}\max_{y_A\in\mathcal{Y}^{\text{card}(Y_A)}}\log \frac{P(Y_A=y_A|Y_i=a,\theta)}{P(Y_A=y_A|Y_i=b,\theta)}.$$
	\label{def:maxInfV}
\end{definition}

%This is insufficient when one wants to protect distributional information of an attribute.
%To this end, Definition~\ref{def:maxInfD} instead measures max-influence
%of a parameter of the probability distribution of a variable on the distribution of parameters of other variables.

Recall the definition of a Markov Quilt (Definition \ref{def.markovquilt}), which is used in this mechanism.

%{\centering
%	\begin{minipage}{\linewidth}
		\begin{algorithm}[H]
		\centering
			\caption{Markov Quilt Mechanism  ($Y, F, \{S, \pufQ, \Theta \}, \epsilon$)  \cite{song2017pufferfish} }
			\begin{algorithmic}
				\State \textbf{Input:} database $Y$, $L$-Lipschitz query $F$, Pufferfish framework $\{S, \pufQ, \Theta \}$, privacy parameter $\epsilon$.
				\For {each $Y_i$}
				\State Let $G_i := \{ (Y_Q, Y_N, Y_R) : Y_Q \text{ is a Markov Quilt of } Y_i \}$
				\For {all $Y_Q$ (with $Y_N, Y_R$ ) in $G_i$}
				\If {$e^v_{\Theta}(Y_Q|Y_i)<\eps$}
				\State Set $b(Y_Q)=\frac{|Y_N|}{\eps-e^v_\Theta(Y_Q|Y_i)}$.
				\Else
				\State Set $b(Y_Q)=\infty$.
				\EndIf
				\EndFor
				\State Set $b_i=\min_{Y_Q\in G_i} b(Y_Q)$.
				\EndFor
				\State Set $b_{\max}=\max_i b_i$.
				\State Sample $Z\sim\Lap(L \cdot b_{\max})$.
				\State Return $F(Y)+Z$
			\end{algorithmic}\label{algo:mq}
		\end{algorithm}
%	\end{minipage}
%}

%\fixme{I left $\sigma$ above. Because $b$ would probably be confusing as we use it as well. May be $d_Q$, $d_i$?}

The Markov Quilt Mechanism  \cite{song2017pufferfish} given in Algorithm \ref{algo:mq} guarantees $(\eps,0)$-Pufferfish Privacy.

%proposed
%in~\cite{song2017pufferfish} is the first
%general mechanism for satisfying instantiations of Pufferfish privacy framework.
%Since it is computationally expensive, Song~\textit{et al.}
%gave the Markov Quilt Mechanism given in Algorithm \ref{algo:mq} for some special structures of data dependence that form a Bayesian network.
%It is more efficient and guarantees $(\eps,0)$-Pufferfish Privacy.

%\begin{definition}[Max-Divergence]
%Let $p$ and $q$ be two distributions with the same support. The max-divergence $D_\infty(p||q)$ between them is defined as:
%$$D_\infty(p||q)=\sup_{x\in \text{support}(p)}\log \frac{p(x)}{q(x)}.$$
%\end{definition}

%
%\wz{Do you still need this?}
%\begin{definition}[Markov Quilt]
%	A set of nodes $X_Q$, $Q\subset \{1,\ldots, n\}$ in a Bayesian network $G=(X,E)$ is a Markov Quilt for a node $X_i$ if the following conditions hold: 
%	\begin{compactitem}
%		\item Deleting $X_Q$ partitions $G$ into parts $X_N$ and $X_R$ such that $X=X_N\cup X_Q \cup X_R$ and $X_i\in X_N$.
%		\item For all $x_R\in \mathcal{X}^{\text{card}(X_R)}$, all $X_Q\in \mathcal{X}^{\text{card}(X_Q)}$ and for all $a\in \mathcal{X}$,
%		$$P(X_R=x_R|X_Q=x_Q,X_i=a)=P(X_R=x_R|X_Q=x_Q).$$
%		Thus, $X_R$ is independent of $X_i$ conditioned on $X_Q$.
%	\end{compactitem}
%\end{definition}

\end{document}